\documentclass{aims}
\usepackage{amsmath}
  \usepackage{paralist}
  \usepackage{graphics} %% add this and next lines if pictures should be in esp format
  \usepackage{epsfig} %For pictures: screened artwork should be set up with an 85 or 100 line screen
\usepackage{graphicx}  \usepackage{epstopdf}%This is to transfer .eps figure to .pdf figure; please compile your paper using PDFLeTex or PDFTeXify.
 \usepackage[colorlinks=true]{hyperref}
\usepackage{xcolor}
\usepackage{lipsum}
\usepackage{amsfonts}
\usepackage{graphicx}
\usepackage{epstopdf}
\usepackage{subfig}
\usepackage{algorithm}
\usepackage{algpseudocode}
\usepackage{amssymb}
\hypersetup{urlcolor=blue, citecolor=red}

\def\currentvolume{X}
 \def\currentissue{X}
  \def\currentyear{200X}
   \def\currentmonth{XX}
    \def\ppages{X--XX}
     \def\DOI{10.3934/xx.xx.xx.xx}

\newtheorem{theorem}{Theorem}[section]
\newtheorem{corollary}{Corollary}

\newtheorem{lemma}[theorem]{Lemma}

\newtheorem{hypothesis}{Hypothesis}
\newtheorem*{hypothesis+}{Hypothesis+}

\newtheorem*{conjecture1}{Conjecture}

\theoremstyle{definition}
\newtheorem{definition}[theorem]{Definition}
\newtheorem{remark}{Remark}

\newtheorem*{example}{Example}

\newcommand\Algphase[1]{%
\vspace*{-.7\baselineskip}\Statex\hspace*{\dimexpr-\algorithmicindent-2pt\relax}\rule{\textwidth}{0.4pt}%
\Statex\hspace*{-\algorithmicindent}\textbf{#1}%
\vspace*{-.7\baselineskip}\Statex\hspace*{\dimexpr-\algorithmicindent-2pt\relax}\rule{\textwidth}{0.4pt}%
}

\DeclareMathOperator{\sgn}{sgn}

\title[Gap solitons: coding and method for computation] %Use the shortened version of the full title
      {Gap solitons for the repulsive Gross-Pitaevskii equation with periodic potential: coding and method for computation}

\author[Georgy L. Alfimov, Pavel P. Kizin and Dmitry A. Zezyulin]{}


 \email{galfimov@yahoo.com}
 \email{p.p.kizin@gmail.com}
 \email{dzezyulin@fc.ul.pt}

\thanks{\color{blue}{The third author is supported by FCT (Portugal) through the grant No. UID/FIS/00618/2013}}

\begin{document}
\maketitle

% Enter the first author's name and address:
\centerline{\scshape Georgy L. Alfimov, Pavel P. Kizin}
\medskip
{\footnotesize
% please put the address of the first author
 \centerline{Moscow Institute of Electronic Engineering, Zelenograd, Moscow, 124498, Russia}
} % Do not forget to end the {\footnotesize by the sign }

\medskip

\centerline{\scshape Dmitry A. Zezyulin}
\medskip
{\footnotesize
 % please put the address of the second  and third author
 \centerline{Centro de F\'{\i}sica Te\'orica e Computacional and Departamento de F\'{\i}sica,}
   \centerline{Faculdade de Ci\^encias da Universidade de Lisboa, Campo Grande, Ed. C8,  Lisboa P-1749-016, Portugal}
}

\bigskip

% The name of the associate editor will be entered by an editorial staff
% "Communicated by the associate editor name" is not needed for special issue.
 \centerline{(Communicated by the associate editor name)}

%The abstract of your paper
\begin{abstract}
The paper is devoted to nonlinear localized modes (``gap solitons'') for the spatially
one-dimensional  Gross-Pitaevskii equation (1D GPE) with a periodic potential and
repulsive interparticle interactions. It has been recently shown (G. L. Alfimov,
A. I. Avramenko, Physica D, 254, 29 (2013)) that under certain conditions all the
stationary modes for the 1D GPE can be coded by bi-infinite sequences of symbols
of some finite alphabet (called ``codes'' of the solutions).  We  present and justify a
numerical method which allows to reconstruct  the profile of a localized mode by its
code. As an example,  the method is applied to compute the profiles of gap solitons
for 1D GPE with a cosine potential.
\end{abstract}

%The title of your section 1
\section{Introduction}\label{Intro}

The Gross--Pitaevskii  equation  (GPE) is a commonly recognized model
for description of nonlinear matter waves in Bose--Einstein
condensates (BECs) \cite{PS03}. In the dimensionless form  spatially one-dimensional (1D) GPE,
\begin{equation}\label{GPE}
  i\Psi_t = -\Psi_{xx} + U(x) \Psi + \sigma |\Psi|^2 \Psi,
 \quad \sigma = \pm 1,
\end{equation}
describes the dynamics of  a cigar-shaped BEC cloud
in the so-called mean-field approximation. In (\ref{GPE}) the variables are properly scaled, $\Psi=\Psi(t,x)$ is the complex-valued macroscopic wavefunction of the condensate, and its squared amplitude $|\Psi( t,x)|^2$ describes the local density of  BEC. The nonlinear term $\sigma |\Psi|^2\Psi$ takes into account the  interactions between the particles: the case $\sigma=1$ corresponds to the repulsive interparticle interactions, $\sigma =-1$ corresponds to the attractive interactions. Both these cases are of physical relevance. The real-valued  potential $U(x)$ describes a trap which confines the BEC. In the case of optical confinement of BEC, the potential $U(x)$ is a periodic function.

An important particular class of solutions of the GPE corresponds to the
stationary modes. These  modes are of the form $\Psi(t, x)
= e^{-i\omega t} \psi(x)$ and the wavefunction $\psi(x)$ solves the
stationary GPE
\begin{equation}
\psi_{xx} + (\omega  - U(x))\psi - \sigma |\psi|^2\psi=0.
\end{equation}
The stationary modes of the condensate that satisfy the localization conditions
\begin{equation}\label{BoundCond}
  \lim_{x\to\pm \infty} \psi(x)=0
\end{equation}
can be regarded as {\it real-valued} \cite{AKS02} and satisfying the 1D real stationary GPE
\begin{equation}\label{eq:main}
  \psi_{xx} + (\omega - U(x))\psi - \sigma\psi^3=0,
  \quad  \psi,x\in \mathbb{R}.
\end{equation}
If the potential $U(x)$ is periodic, the modes satisfying (\ref{BoundCond})
and (\ref{eq:main}) are called \emph{gap solitons}. This terminology can
be understood if one considers the linear version of (\ref{eq:main}), i.e.
\begin{equation}\label{eq:linear}
  \psi_{xx} + (\omega - U(x))\psi =0.
\end{equation}
Due to the conditions (\ref{BoundCond}), (\ref{eq:linear})
approximates (\ref{eq:main}) for $|x|\gg 1$. Equation (\ref{eq:linear})
is a linear eigenvalue problem for the Schr\"odinger operator with periodic
potential $U(x)$ and  eigenvalue $\omega$. It is well-known \cite{BS91},
that the spectrum of such a problem typically consists of a  countable  set
of continuous \emph{bands}  separated by \emph{gaps}. Zero boundary
conditions at $x\to+\infty$ and $x\to -\infty$ for (\ref{eq:linear}) cannot
be satisfied  if $\omega$ lies in a band of the spectrum.  Therefore the
solutions of (\ref{eq:main}) satisfying (\ref{BoundCond}) can
exist only if $\omega$ belongs to one of the gaps.

\begin{figure}
  \centerline{\includegraphics [scale=0.8]{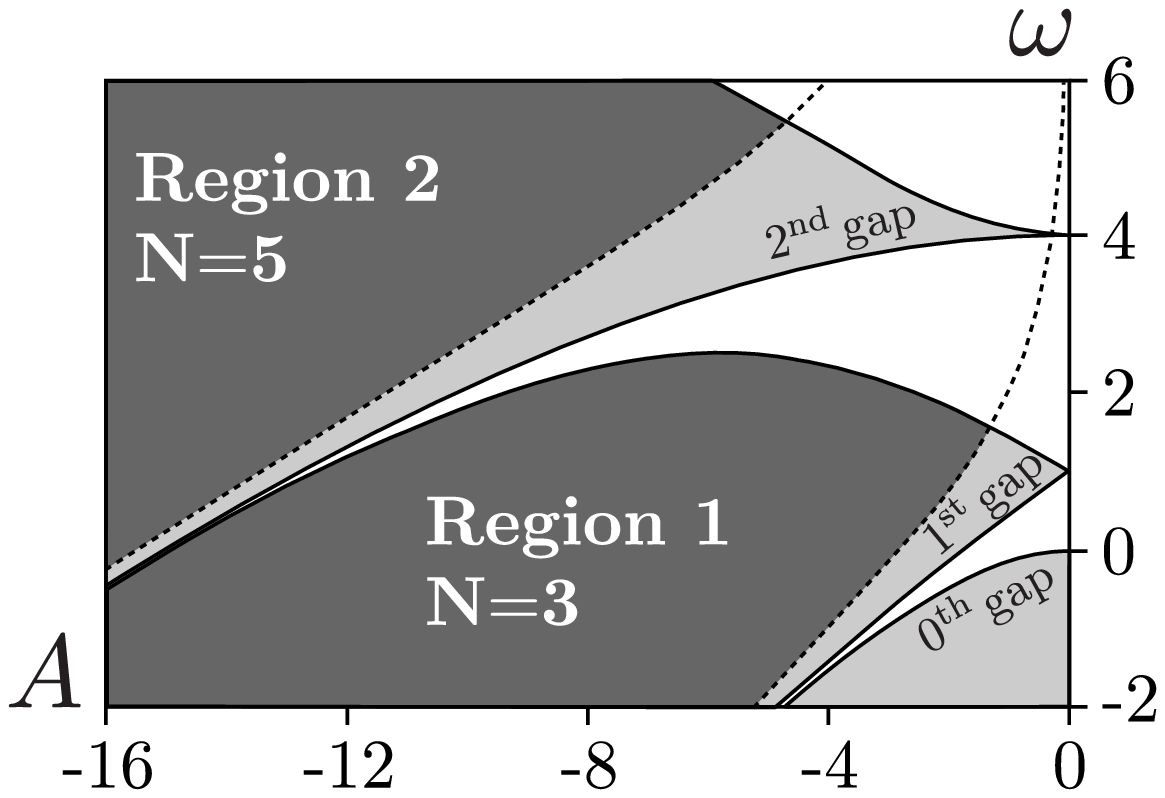}}
  \caption{The diagram of gaps and bands of the Matheiu equation
   $\psi_{xx}+(\omega-A\cos 2x)\psi=0$. The regions where
    according to \cite{AA13} the coding of bounded solutions for
    the nonlinear equation (\ref{eq:main}) is possible are shown in the first and the second
    gaps: if $\omega$ and $A$ belong to Region~1, then the bounded
    solutions can be coded using an alphabet of
    $N=3$ symbols; in Region~2  an alphabet of $N=5$ symbols is
    necessary.}\label{CodZones}
\end{figure}

It is known that apart from the gap solitons \cite{AKS02, Dark, Kiv2003,
PSK04}, (\ref{eq:main}) supports nonlinear periodic structures
(nonlinear Bloch waves) \cite{Kiv2003,BlochW}, domain walls
\cite{DomWalls05}, gap waves \cite{Kiv2006} and so on. The
authors of \cite{Zhang09_2,Zhang09_1} observed that more complex
structures described by (\ref{eq:main}), such as nonlinear Bloch waves, can be
built from elementary entities called \emph{fundamental gap solitons}
(FGS)  by means of the so-called \emph{composition relation}. This idea was
further developed in \cite{AA13}, where it was suggested to code
complex nonlinear modes of (\ref{eq:main}) using the methods of
symbolic dynamics. The approach of  \cite{AA13} is consistent with
the recent results of  \cite{FS14} where it was proved that in the semiclassical limit a certain  class of gap solitons for (\ref{eq:main}) can be mapped
to a set of  localized modes of Discrete Nonlinear Schr\"odinger Equation
(DNLS) which admit  the description in terms of coding sequences %is possible
\cite{ABK04}.

The ``coding'' approach of \cite{AA13} exploits the following fact.  In the case of repulsive interactions
($\sigma=1$)  the ``most'' of the solutions of  (\ref{eq:main})
\emph{collapse} (i.e., tend to infinity) at some finite point of the real axis.
The complementary set of non-collapsing solutions (including all gap soliton solutions)
can be associated with a set on the
plane of initial conditions for (\ref{eq:main}). This set  has fractal structure and under certain  conditions  (Hypotheses 1-3 of  \cite{AA13}) it
can be described in terms of symbolic dynamics and mapped into the set
of bi-infinite sequences of symbols of some finite  alphabet $\mathcal{A}$.
The relation between the solutions and the codes is a homeomorphism:
each bounded in $\mathbb{R}$ solution of (\ref{eq:main})
corresponds to  unique bi-infinite code of the type
\begin{displaymath}
  \{\ldots, i_{-1},i_0,i_1,\ldots\}, \quad i_k\in {\mathcal{A}},
\end{displaymath}
and each code of this type corresponds to unique bounded in $\mathbb R$
solution of (\ref{eq:main}). It was argued in \cite{AA13}, that in the case of the cosine potential
\begin{equation}\label{CosPot}
  U(x)=A\cos 2x
\end{equation}
this coding is possible for vast areas in the parameter plane $(\omega,A)$.
These areas are situated in the gaps of linear equation (\ref{eq:linear}) which  becomes the Mathieu equation for the cosine potential (\ref{CosPot}) (see Figure~\ref{CodZones}).  For  Region~1 (situated in the first gap),
the alphabet ${\mathcal{A}}$ consists  of three symbols,  while in the Region~2
(situated in the second gap of the Mathieu equation) the alphabet of five symbols must be used.

In order to transform the coding approach  into a convenient tool for   analysis and computation of gap solitons, one has to answer the following practical questions:
\medskip

{\sc Question 1.} How can one get a code of a given gap soliton?

{\sc Question 2.} How can one compute a profile of gap soliton by its code?
\medskip

The answer to the Question~1 is relatively simple (see subsection~\ref{GenRem}) and  follows immediately from the results of
\cite{AA13}. The answer to the
Question~2 is not simple, and it is the focal point of  this paper. From the practical point of view, the answer to this question offers a method for numerical computation of gap solitons. While several numerical approaches to this task have been described in the literature
(see, for instance, a survey in the monograph \cite{Yang10}), the peculiar advantage of a coding-based numerical algorithm is that it yields  a nonlinear mode which corresponds to the specific  code chosen beforehand. Applying the algorithm to different codes it is possible to compute and classify gap solitons with  shapes of different    complexity assigned in advance.

{
Looking ahead, we note that a high-accuracy algorithm for computation of
nonlinear modes is a crucial ingredient of the numerically accurate  stability analysis. It is known that even the simplest FGSs from lower gaps undergo  oscillatory
instabilities under both  attractive \cite{PSK04}  and repulsive    \cite{KZA16} nonlinearities (instabilities of this type have been also found in the coupled mode limit \cite{BarPelZem98,BarZem00,Malomed94}).
The oscillatory instabilities of gap solitons can be extremely  weak and
can be revealed by a quite thorough study only \cite{KZA16}. This, in particular, requires to compute the profile of nonlinear mode with    high accuracy.}

In what follows, the consideration is restricted by the repulsive version of
GPE ($\sigma=1$). Section~\ref{Coding} contains a brief description of
the coding approach. Section~\ref{Structure} includes some theoretical results
necessary for construction of the algorithm. Section~\ref{subAlgorithm} contains the
description of the main numerical algorithm whose listing   is presented in Appendix~\ref{App_D}.
In section~\ref{Results}, the method is illustrated for the GPE with cosine potential
(\ref{CosPot}). Section~\ref{Disc} concludes the paper and  presents a short outlook
for the future work.

\section{Review of the coding approach}
\label{Coding}

%\subsection{Basic definitions and assumptions}
%\label{Definitions}

Let us briefly describe the method for coding of nonlinear modes for
(\ref{eq:main}) (see \cite{AA13} for the detailed presentation).  The basis of this method, the coding theorem, is formulated in subsection~\ref{sec:codth}. In order to formulate the coding theorem, in subsection~\ref{sec:collapse} and \ref{sec:geo} we introduce several definitions, assumptions, and auxiliary concepts.

\subsection{Collapsing solutions}
\label{sec:collapse}
For  $\sigma=1$ the equation (\ref{eq:main}) for the nonlinear modes reads
\begin{equation}\label{1D_rep}
  \psi_{xx} + (\omega - U(x))\psi - \psi^3=0, \quad
  \psi,x\in \mathbb{R}.
\end{equation}
The potential $U(x)$ is assumed to be bounded, $\pi$-periodic and even.
%\begin{displaymath}
%  U(x)=U(x+\pi),\quad U(x)=U(-x)
%\end{displaymath}
A prototypical example is the cosine potential (\ref{CosPot}).
% that has been discussed in many studies.
The following definitions are necessary,
\cite{AA13}.
\begin{definition}
  The solution $\psi(x)$ of Cauchy problem $\psi(0)=\psi_0$, $\psi_x(0)=\psi_0'$ of (\ref{1D_rep}) collapses to $+\infty$ (respectively, collapses to $-\infty$) at the point $x=x_0>0$ if  $ \lim_{x\to x_0-0}\psi(x)=+\infty$, (respectively, $ \lim_{x\to x_0-0}\psi(x)=-\infty$). Similarly, the solution $\psi(x)$ of this Cauchy problem collapses to $+\infty$ (respectively, collapses to $-\infty$) at the point $x=x_0<0$ if $ \lim_{x\to x_0+0}\psi(x)=+\infty$ (respectively, $ \lim_{x\to x_0+0}\psi(x)=-\infty$).
  %Alternatively, one says that the solution $\psi(x)$ {\it collapses} at $x_0$.
\end{definition}

The balance of dispersive and nonlinear terms in (\ref{1D_rep}) allows to conclude that in vicinity of singularity $x=x_0$ the asymptotics of a collapsing solution is
\begin{equation*}
\psi(x)\sim\pm \frac{\sqrt{2}}{x-x_0}.
\end{equation*}

\begin{definition}\label{def1}
  Let $(\psi_0,\psi_0')$ be a point on the plane $\mathbb{R}^2=(\psi,\psi')$.
  Let $\psi(x)$ be a solution of the Cauchy problem for (\ref{1D_rep}) with ininital data $\psi(0)=\psi_0$,  $\psi_x(0)=\psi_0'$. Then the point $(\psi_0,
  \psi_0')$ is called
  \begin{itemize}
  \item $P$-collapsing to $+\infty$ (or $-\infty$), if the solution
  	$\psi(x)$ collapses to $+\infty$ (correspondingly, to $-\infty$) at the point $x=P$;
  \item $P$-non-collapsing forward point, if $\psi(x)$ does not collapse at any
    $x\in(0;P]$;
  \item $P$-non-collapsing backward point, if $\psi(x)$ does not collapse at any $x\in[-P;0)$.
  \end{itemize}
\end{definition}

\begin{example}
Let $U(x)\equiv 0$ and $\omega=0$. Then:

a. The solution of Cauchy problem for (\ref{1D_rep}) and
$\psi(0)=\psi_x(0)=\sqrt{2}$ is
\begin{equation*}
\psi(x)=\frac{\sqrt{2}}{1-x}.
\end{equation*}
Then $\psi(x)$ collapses to $+\infty$ at $x_0=1$. The point $(\sqrt{2},\sqrt{2})$ is $1$-collapsing to $+\infty$ point. Simultaneously, it is $\infty$-non-collapsing backward point, since no collapse occurs at the interval $(-\infty;0)$.

b. The solution of Cauchy problem for (\ref{1D_rep}) and
$\psi(0)=1$, $\psi_x(0)=0$ is
\begin{eqnarray*}
\psi(x)=\frac{1}{{\rm cn}(x;\sqrt{2}/2)}.
\end{eqnarray*}
Here ${\rm
cn}~(x;k)$ is the Jacobi elliptic function.
Then $\psi(x)$ collapses to $+\infty$ at $x_\pm=\pm K(\sqrt{2}/2)$, where $K(\xi)$ is the complete elliptic integral of the first kind. The point $(1,0)$ is $K(\sqrt{2}/2)$-collapsing to $+\infty$ point and $-K(\sqrt{2}/2)$-collapsing to $+\infty$ point. Also, since $K(\sqrt{2}/2)\approx 1.854>1$ one can say that
the point $(1,0)$ is $1$-non-collapsing forward point and $1$-non-collapsing backward point.
\end{example}
\medskip

\begin{definition}\label{def2}
  For (\ref{1D_rep}) denote by ${\mathcal{U}}_P^+$ the set of all $P$-non-collapsing forward points and by ${\mathcal{U}}_P^-$ the set of all $P$-non-collapsing backward points.
\end{definition}

\begin{figure}
  \includegraphics [width=0.9322\textwidth]{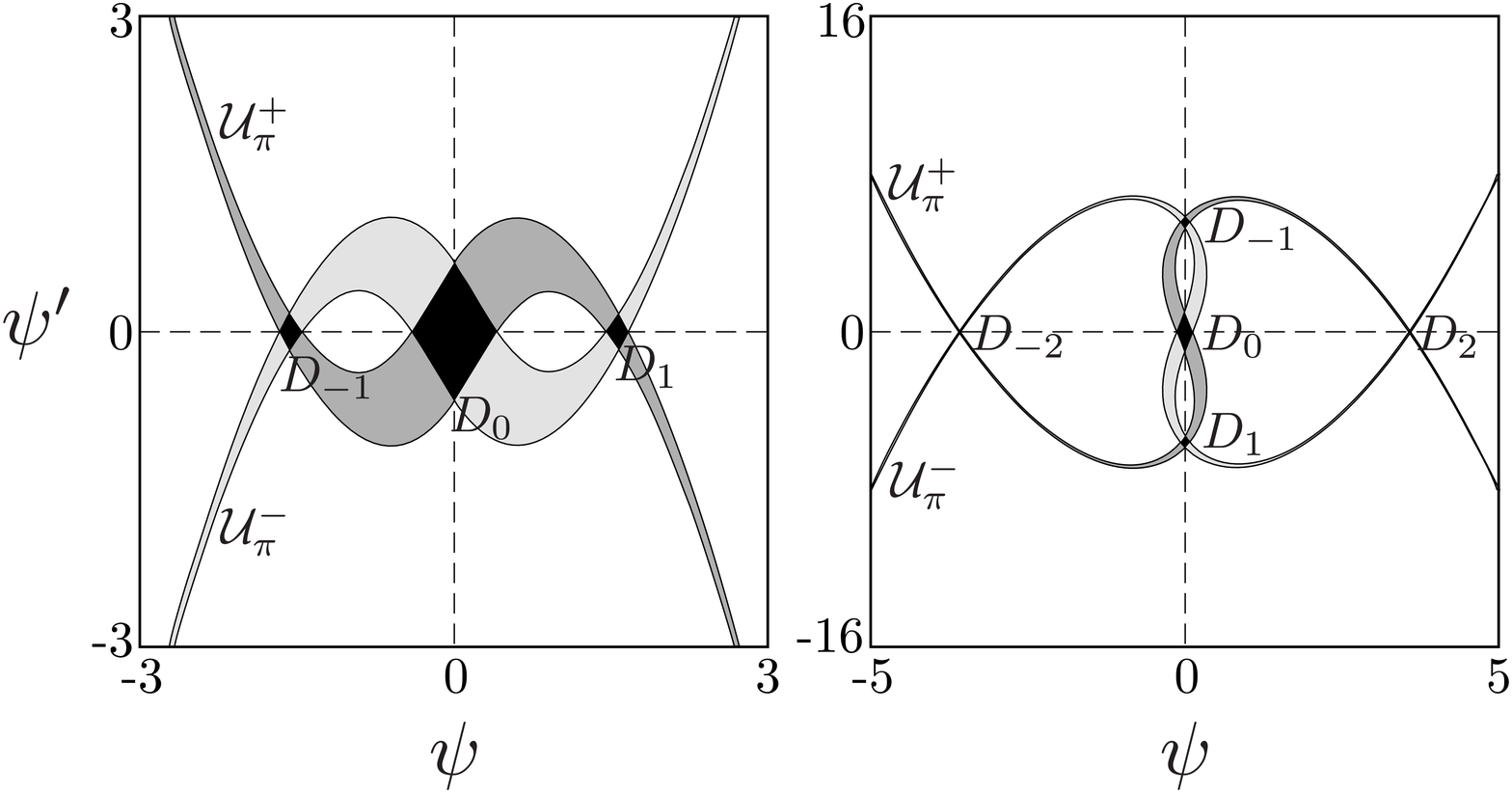}
  \caption{The sets ${\mathcal{U}}_\pi^\pm$ and $\Delta_0$ for (\ref{1D_rep}) with the potential (\ref{CosPot}). Panel (a): parametric point $(\omega,A)=(1,-3)$ from Region 1 in Figure~\ref{CodZones}, in this case $\Delta_0$ consists of three connected components, $D_{-1}$, $D_0$ and $D_{1}$; Panel (b): parametric point $(\omega,A)=(4,-10)$ from Region 2 in Figure~\ref{CodZones}, in this case $\Delta_0$ consists of five connected components, from $D_{-2}$ to $D_2$.}
  \label{U_pi}
\end{figure}

\begin{definition} Define the Poincare map $T:\mathbb{R}^2\to\mathbb{R}^2$ associated
  with (\ref{1D_rep}) as follows: for a point $(\psi_0,\psi_0')\in \mathbb{R}^2$,  $T(\psi_0,\psi_0')=(\psi(\pi),\psi_x(\pi))$ where
  $\psi(x)$ is a solution of the Cauchy problem for (\ref{1D_rep}) with
  initial data $\psi(0)=\psi_0$ and $\psi_x(0)=\psi_0'$.
\end{definition}
The Poincare map  $T$ and the sets ${\mathcal{U}}_\pi^\pm$ for (\ref{1D_rep}) have the following properties:
\begin{enumerate}
  \item $T$ is an area-preserving diffeomorphism.
  \item $T$ is defined on the set ${\mathcal{U}}_\pi^+$ only. Correspondingly,
    inverse map $T^{-1}$ is defined on the set ${\mathcal{U}}_\pi^-$. Moreover, $T{\mathcal{U}}_\pi^+={\mathcal{U}}_\pi^-$ and $T^{-1}{\mathcal{U}}_\pi^-={\mathcal{U}}_\pi^+$.

  \item Since $U(x)$ is an even function, the sets $\mathcal{U}_\pi^\pm$ are symmetric with respect to $\psi$.

  \item Since the nonlinearity in (\ref{1D_rep}) is an odd function of $\psi$,
    the sets ${\mathcal{U}}_\pi^\pm$ are symmetric with respect to the origin (0,0).

  \item Since $\mathcal{U}_\pi^\pm$ are symmetric with respect to $\psi$ and $(0,0)$,
  they are symmetric with respect to $\psi'$.
  \item The boundary of ${\mathcal{U}}_\pi^+$  (to be denoted by $\partial
    {\mathcal{U}}_\pi^+$)  consists of $\pi$-collapsing to $+\infty$ or
    $-\infty$ points. This boundary is formed by continuous curves (Corollary of Theorem 2.2 in \cite{AA13}, see also \cite{UMJ16}) and
    is symmetric with respect to the origin. Moreover, if some point
    $(\psi,\psi')$ of $\partial {\mathcal{U}}_\pi^+$ is $\pi$-collapsing to
    $+\infty$, then the symmetric point $(-\psi,-\psi')$ also belongs to
    $\partial {\mathcal{U}}_\pi^+$ and it is $\pi$-collapsing  to $-\infty$.
    The similar situation takes place for $\partial {\mathcal{U}}_\pi^-$, the
    boundary of ${\mathcal{U}}_\pi^-$. It consists of $-\pi$-collapsing to $+\infty$ of $-\infty$ points.
\end{enumerate}

\begin{figure}
  \includegraphics [width=\textwidth]{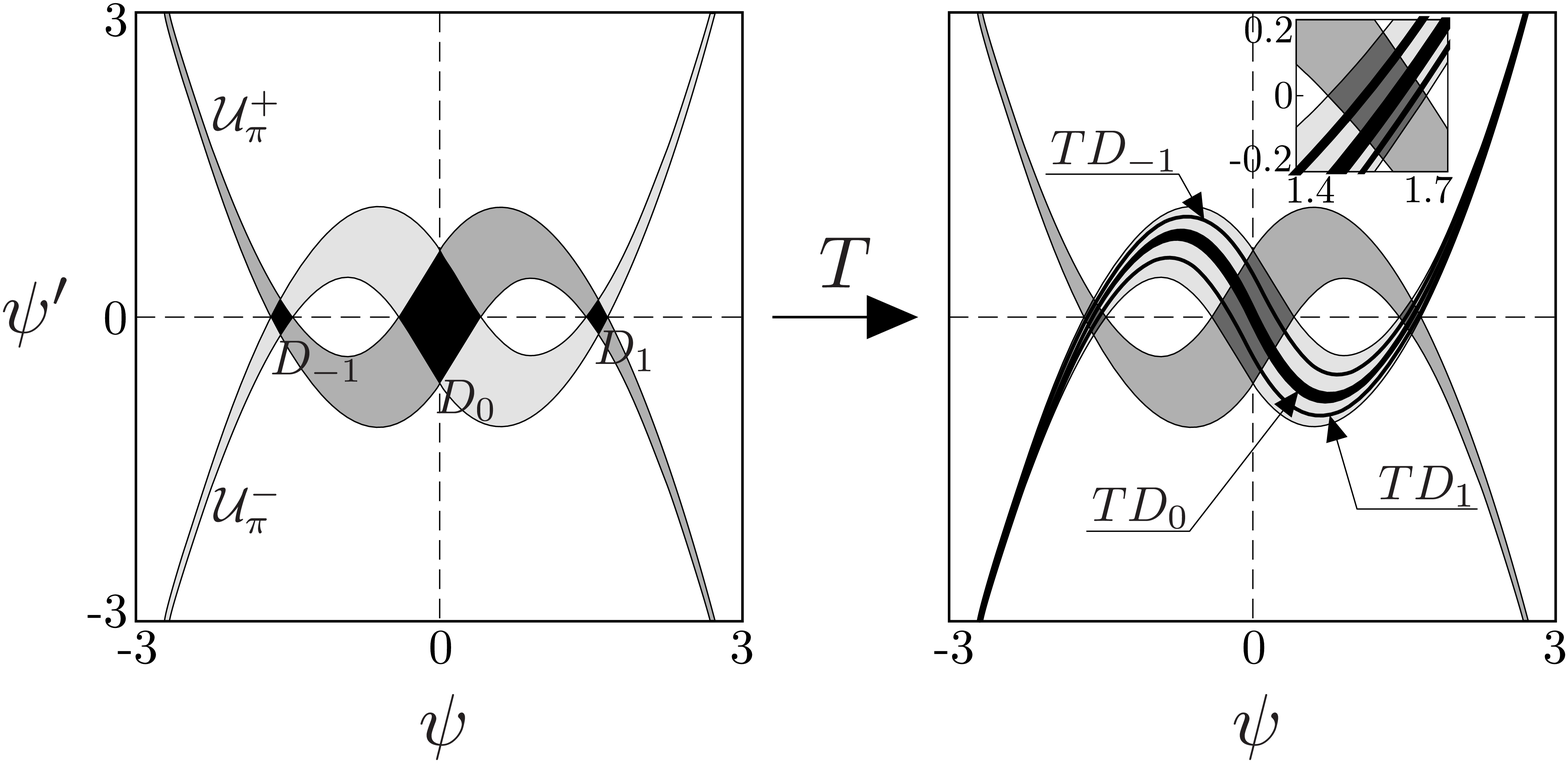}
  \caption{Action of the map $T$ on $\Delta_0$ which consists of three connected components, $D_{-1}$, $D_{0}$ and $D_{1}$. The cosine potential (\ref{CosPot}) was used for (\ref{1D_rep}). Parameters are $\omega=1$ and $A=-3$.}
  \label{Action}
\end{figure}

\begin{definition}
  Denote
  \begin{equation}\label{D0}
  \Delta_0 = \mathcal{U}_{\pi}^+ \cap \mathcal{U}_{\pi}^-.
  \end{equation}
\end{definition}
Evidently, $\Delta_0$ consists of the points that have both $T$-image
and $T$-pre-image. The features of $\Delta_0$ are as follows:

\begin{enumerate}
  \item $\Delta_0$ is bounded (Theorem 2.1 in \cite{AA13})
  and includes the origin $(0,0)$.

  \item $\Delta_0$ is open (it follows from Theorem 2.2 in \cite{AA13}).

  \item If $\Delta_0$ consists of the finite number $N$ of connected
  components, then  $N$ is {\it odd} ($N\geq 3$), and for each
  connected component $D\subset\Delta_0$ and $(0,0)\notin D$ one can
  find a dual component $D^*\subset\Delta_0$ such that $D$ and $D^*$ are situated symmetrically with respect to the origin.

\end{enumerate}

\begin{definition}
  The sets $\Delta_n^\pm$, $n=0,1,\ldots$ are defined by the following
  recurrence rule
  \begin{align}
    &\Delta_0^-=\Delta_0,\quad \Delta^-_{n+1}=
      T\Delta_{n}^-\cap\Delta_0, 	\label{D-n} \\
    &\Delta_0^+=\Delta_0,\quad \Delta^+_{n+1}=
      T^{-1}\Delta_{n}^+\cap\Delta_0.    \label{D+n}
  \end{align}
\end{definition}
Evidently, if $p\in \Delta_{n}^+$ then there exist $Tp\in
\Delta_{n-1}^+\subseteq \Delta_0$, $T^2p\in  \Delta_{n-2}^+
\subseteq \Delta_0$, $\ldots$, $T^np\in \Delta_{0}$. On the other hand,
if for a point $p\in \Delta_0$ there exists $q=T^n p\in\Delta_0$, then
$T^{-1}q=T^{n-1}p\in \Delta_1^{+}$, $T^{-2}q=T^{n-2}p\in
\Delta_2^{+}$, etc. Repeating the procedure $n$ times one arrives at the
relations
\begin{align*}
  &\{p\in \Delta^{+}_{n}\} \iff
    \{p,Tp,T^2p,\ldots,T^n p\in\Delta_0\}, \\
  &\{p\in \Delta^{-}_{n}\} \iff
    \{p,T^{-1}p,T^{-2}p,\ldots,T^{-n} p\in\Delta_0\}.
\end{align*}
These relations imply that
\begin{align*}
  &\ldots \subset \Delta_{n+1}^+\subset \Delta_{n}^+\subset
    \ldots \subset
    \Delta_{1}^+\subset \Delta_{0},\\
  &\ldots \subset \Delta_{n+1}^-\subset \Delta_{n}^-\subset
    \ldots \subset
    \Delta_{1}^-\subset \Delta_{0}
\end{align*}
Let us illustrate the construction of the sets $\Delta_n^\pm$ by an
example. Figure~\ref{Action} shows the action of $T$ on $\Delta_0$
for the case of potential $U(x)=-3\cos 2x$ and $\omega=1$. The set
$\Delta_0$ consists of three connected components, marked  $D_{-1}$,
$D_0$ and $D_1$. $T$-images of these components are
three infinite strips. Each of these infinite strips intersects each of the connected components $D_{-1,0,1}$, and therefore the set $\Delta^-_1=T\Delta_0\cap\Delta_0$
consists of nine connected components. Further, the set $\Delta_2^-$ consists of
27 connected components and so on. The sets $\Delta_n^+$ are symmetric
to $\Delta_n^-$ with respect to $\psi$-axis, so they have the same
properties.

\subsection{Assumptions}\label{GeomCon}

%\subsection{The coding theorem}
\label{sec:geo}

Recall  that a function $f(x)$ is called
%{\it $\varkappa$-Lipschitz	function}
{\it Lipschitz function} if there exists finite $\varkappa>0$ such that for any $x_1$ and $x_2$ one has
$|f(x_2)-f(x_1)|\leq\varkappa|x_2-x_1|$.
% the relation holds
%\begin{displaymath}
%|f(x_2)-f(x_1)|\leq\varkappa|x_2-x_1|.
%\end{displaymath}
In order to
formulate the main result of \cite{AA13} let us introduce the following
definitions.

\begin{definition}\label{def:incrdecr}
Let us call a plane curve $\gamma$ on the plane $(\psi,\psi')$ {\rm increasing} (correspondingly, {\rm decreasing}), if $\gamma$ is a graph of \textcolor{black}{increasing}
(correspondingly, decreasing) Lipschitz function  $\psi'=G(\psi)$.
\end{definition}

\begin{definition}[An island]\label{def:island1}
%Let $\varkappa$ be a fixed real.
We call an island
an open curvilinear quadrangle $D\subset \mathbb{R}^2$ bounded  by two pairs of curve segments, such that the segments from one  pair do not intersect and are increasing, and the segments from another  pair do not intersect and are decreasing.
%nonintersecting  curve segments $\alpha^+$, $\beta^+$, $\alpha^-$,
%$\beta^-$. The bounds $\alpha^+$ and $\alpha^-$ are opposite to each other and are both non-increasing or non-decreasing. Also $\beta^+$ and $\beta^-$ are both non-increasing or non-decreasing and if $\alpha^\pm$ are  non-increasing then $\beta^\pm$ are non-decreasing or vice-versa.
\end{definition}

Let us  make the following assumptions about the set $\Delta_0$ and
the map $T$ associated with (\ref{1D_rep}).\medskip

\begin{hypothesis}\label{hyp1}
The set $\Delta_0$ defined by (\ref{D0}) consists of $N=2L+1$  disjoined islands
$D_i$, $i=-L,\ldots,L$.
\end{hypothesis}
%Let us denotes the boundaries of $i$-th island $\alpha^+_i$, $\alpha^-_i$,
%$\beta^+_i$ and $\beta^-_i$. This notation is consistent with the notation
%in definition of an island, see Sect.~\ref{GeomCon}.

We say that one boundary of an island is opposite to another boundary if these two boundaries do not intersect. \medskip

\begin{hypothesis+}\label{hyp+}
For each island $D_i$ from Hypothesis~\ref{hyp1}, one boundary of $D_i$ consists of $\pi$-collapsing to $+\infty$ points, and the opposite   boundary consists of $\pi$-collapsing to $-\infty$ points. Similarly, from  another pair of opposite  boundaries, one consists of  $-\pi$-collapsing to $+\infty$ points, and another   consists of  $-\pi$-collapsing to $-\infty$ points.
\end{hypothesis+}

In what follows, for an island $D_i$, we will denote $\alpha^\pm_i$ the boundaries consisting of  $\pi$-collapsing to $\pm\infty$ points and denote $\beta^\pm_i$ the boundaries consisting of  $-\pi$-collapsing to $\pm\infty$ points, see Figure~\ref{Island(def)}.

Hypothesis\hyperref[hyp+]{+} naturally holds when ${\mathcal{U}}_\pi^\pm$ are infinite
curvilinear strips (see Figure~\ref{U_pi} as an example). In this case the
boundary $\partial {\mathcal{U}}_\pi^+$ is formed by two curves, one of them consists of $\pi$-collapsing
to $+\infty$  points and other one consists of $\pi$-collapsing to
$-\infty$ points.  These  curves can be mapped one into another  by  reflection with respect to the origin. The similar situation takes place for  the boundary
$\partial {\mathcal{U}}_\pi^-$. The opposite sides $\alpha^\pm$ of each
island are formed by segments of opposite boundaries of $\partial {\mathcal{U}}_\pi^+$,
as well as opposite sides $\beta^\pm$ are formed by
segments of opposite boundaries of $\partial {\mathcal{U}}_\pi^-$.

\begin{definition}[$v$- and $h$- curves]\label{def8}
%Let $\varkappa$ be a fixed real and
Let
$D$ be one of the  islands from Hypothesis~\ref{hyp1} bounded by curve segments $\alpha^\pm$ and  $\beta^\pm$.
We call $v$-curve a Lipschitz curve $\beta$ with endpoints on
$\alpha^-$ and $\alpha^+$ which is increasing if $\beta^\pm$ are increasing and is decreasing if $\beta^\pm$ are decreasing. Similarly,
we call $h$-curve a Lipschitz curve $\alpha$ with endpoints on
$\beta^-$ and $\beta^+$ which is increasing if $\alpha^\pm$ are increasing and is decreasing if $\alpha^\pm$ are decreasing.
\end{definition}

\begin{figure}
\centering
\includegraphics [width=0.5\textwidth]{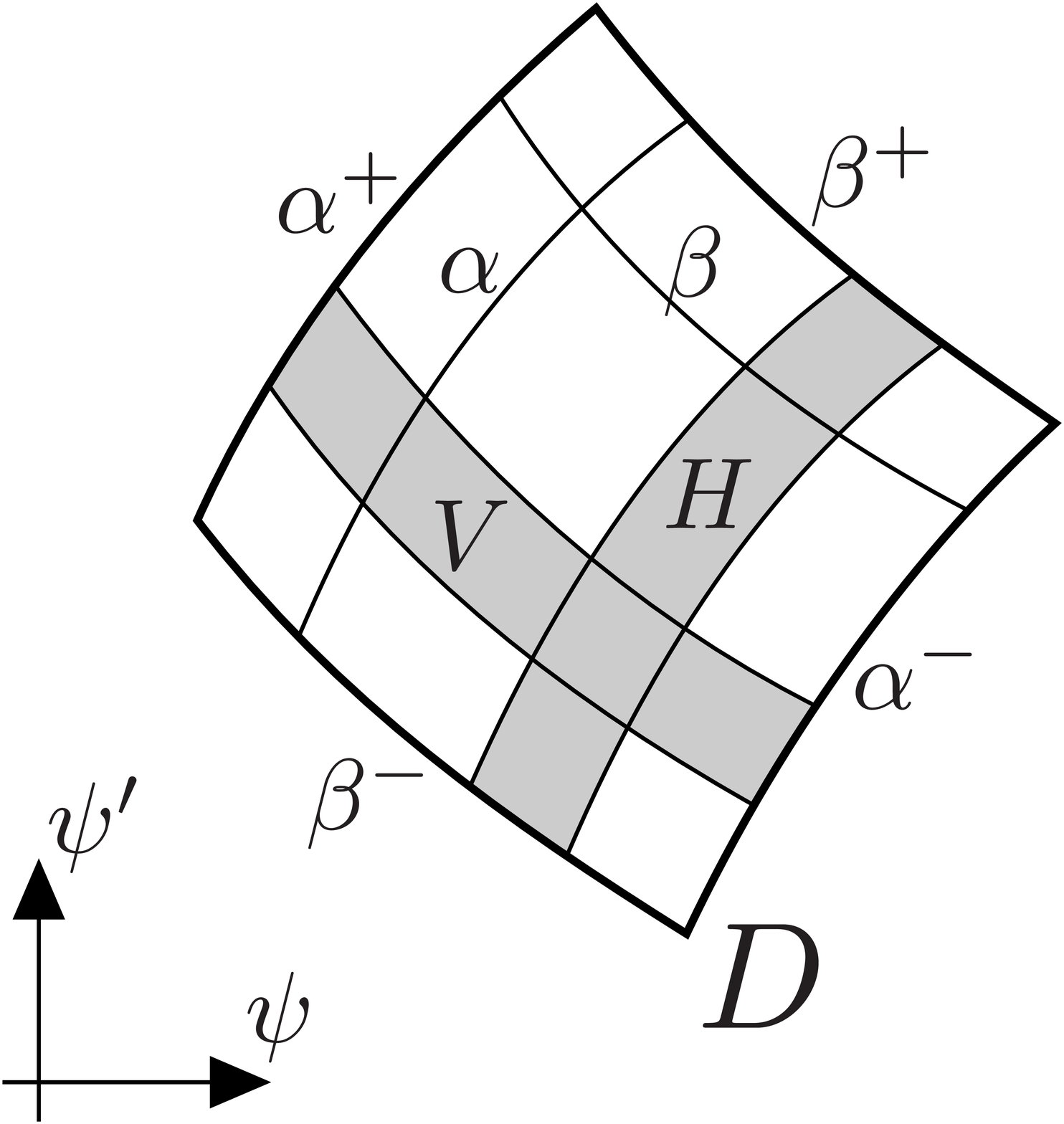}
\caption{An island $D$ with $v$-curve $\beta$, $v$-strip $V$, $h$-curve
$\alpha$ and $h$-strip $H$.}
\label{Island(def)}
\end{figure}

\begin{definition}[$v$- and $h$- strips]
Let $D$ be an island. We call $v$-strip an open curvilinear strip
 between two non-intersecting $v$-curves. Similarly, we call
$h$-strip an open curvilinear strip  between two
non-intersecting $h$-curves.
\end{definition}

\begin{hypothesis}\label{hyp2}
The Poincare map $T$ associated with (\ref{1D_rep}) is such that
	\begin{itemize}
		\item[(a)] $T$ maps $v$-strips of any island $D_i$, $i=-L,\ldots,L$ from $\Delta_0$ in such a
		way that for any $v$-strip $V$, $V\subseteq D_i$, each of the intersections
		$TV\cap D_j$, $j=-L,\ldots,L$, is nonempty and is a $v$-strip.
		\item[(b)] $T^{-1}$ maps $h$-strips of any $D_i$, $i=-L,\ldots,L$,
		in such a way that for any $h$-strip $H$, $H\subseteq  D_i$, each of the
		intersections $T^{-1}H\cap D_j$, $j=-L,\ldots,L$, is nonempty
		and is an $h$-strip.\medskip
	\end{itemize}
\end{hypothesis}

\begin{hypothesis}\label{hyp3}
The sets $\Delta_n^\pm$ defined by (\ref{D-n}), (\ref{D+n}) are such that
	\begin{displaymath}
	\lim_{n\to\infty}\mu(\Delta_n^\pm)=0,
	\end{displaymath}
	where $\mu(S)$ is the  area of $S$.
\end{hypothesis}

For any particular choice of potential $U(x)$ and parameter $\omega$ the
validity of Hypotheses~\ref{hyp1}--\ref{hyp3} should be supported by numerical arguments. On the basis of a systematic numerical investigation,  the following conjecture was
formulated  in \cite{AA13}:\medskip

\begin{conjecture1}\label{conj}
If the parameters $\omega$ and $A$ for (\ref{1D_rep}) with the  cosine potential (\ref{CosPot}) belong to dark gray areas (Regions 1 and 2)  on the parameter plane $(\omega,A)$  in Figure~\ref{CodZones}, then   Hypotheses~\ref{hyp1}--\ref{hyp3}  hold.
\end{conjecture1}

 The upper boundaries of Regions 1 and 2 coincide with upper boundary
of the first and the second gap, respectively. The lower boundaries of these
regions correspond to breaking of Hypothesis~\ref{hyp1}: if $\omega$ and $A$ lie below these boundaries, then $\Delta_0$ is not split into a set of disjoint  islands, i.e., Hypothesis~\ref{hyp1} does not hold. Hypothesis\hyperref[hyp+]{+} was not introduced in \cite{AA13}. However for  $\omega$ and $A$ lying in Regions 1 and 2  the sets ${\mathcal{U}}_\pi^\pm$ are
curvilinear strips bounded by two curves, consisting of $\pm\pi$-collapsing
to $+\infty$  points and of $\pm\pi$-collapsing to
$-\infty$ points correspondingly. This implies that Hypothesis\hyperref[hyp+]{+} also holds for
$\omega$ and $A$ lying in Regions 1 and 2.

%\ref{Island(def)} illustrates schematically the definitions
%introduced above.

\subsection{Coding theorem}
\label{sec:codth}

\begin{definition}
 Denote by  $\Omega^N$, $N=2L+1$, the set of bi-infinite sequences
  $\{\ldots,i_{-1},i_0,i_1,\ldots\}$, $i_k=-L,\ldots,L$.
\end{definition}
The set  $\Omega^N$ has the structure of a  topological space where the
neighborhood of a point $a^*=\{\ldots,i^*_{-1},i^*_0,i^*_1,\ldots\}$
is defined by the sets
\begin{displaymath}
  W_k(a^*)=\{a\in \Omega^N|~ i_j=i^*_j,|j|<k\},\quad k=1,2,\ldots
\end{displaymath}
%Let  $T$ be the Poincare map for \ref{1D_rep} defined on a set
%$\Delta_0=\bigcup_{i=1}^N D_i$ where each $D_i\subset \mathbb{R}^2$,
%$i=-L,\ldots,L$, is an island and all the islands are disjoined.

\begin{definition}
  Denote by $\mathcal{P}$ the set of bi-infinite sequences (called orbits)
  \begin{displaymath}
    {\bf s}=\{\ldots,p_{-1},p_0,p_1,\ldots\},\quad Tp_n=p_{n+1},
  \end{displaymath}
  where each $p_n=(\psi_n,\psi'_n)$, $n=0,\pm1,\pm 2,\ldots$, belongs to
  $\Delta_0$.
\end{definition}
It is clear that  any orbit is uniquely determined by its single entry, for
instance, $p_0$. The point $p_0$ cannot be arbitrary since, as it was mentioned before,
a great part of points of the plane of initial data are collapsing, either
forward or backward (or both), and a collapsing point   cannot generate a
bi-infinite sequence ${\bf s}$. The set $\mathcal{P}$ has the structure of a metric
space where the distance $\rho$ between the elements ${\bf s}^{(1)}=
\{\ldots,p_{-1}^{(1)},p_0^{(1)},p_1^{(1)},\ldots\}$ and ${\bf
s}^{(2)}=\{\ldots,p_{-1}^{(2)},p_0^{(2)},p_1^{(2)},\ldots\}$ is
defined as Euclidean distance between the points $p_0^{(1)}$ and
$p_0^{(2)}$ in $\mathbb{R}^2$,
\begin{displaymath}
  \rho({\bf s}^{(1)},{\bf s}^{(2)})=\sqrt{\left(\psi_0^{(2)}-
  \psi_0^{(1)}\right)^2+\left(\psi_0'^{(2)}-\psi_0'^{(1)}\right)^2}.
\end{displaymath}

%\subsubsection{Hypotheses and the coding theorem}
%\label{CodTheorem}

\begin{definition}\label{def:sigma}
  Define a map $\Sigma:{\mathcal{P}}\to \Omega^N$ as follows: $i_k$ is the
  number $i$ of the component $D_i$ where the point $p_k$ lies.
\end{definition}

The following statement was proved in \cite{AA13}:

\begin{theorem}\label{theorem1}
 Assume that Hypotheses~\ref{hyp1}--\ref{hyp3} hold. Then $\Sigma$ is a
  homeomorphism between the topological spaces $\mathcal{P}$ and $\Omega^N$.
\end{theorem}

In the rest of the paper we assume that Hypotheses~\ref{hyp1}--\ref{hyp3} and Hypothesis\hyperref[hyp+]{+} hold.

\section{Structure of the set $\Delta_0$}
\label{Structure}

Consider the structure of the set $\Delta_0$ in details.
%Assume that Hypotheses 1-3 and Hypothesis+ hold.
{According to the made assumptions,} $\Delta_0$ consists
of $N$ islands, $N=2L+1$, indexed as $D_{-L},\ldots, D_L$. By Theorem~\ref{theorem1}, there exists a homeomorphism between the set of bi-infinite orbits and  ``codes'' of
the form $\{\ldots,i_{-1},i_0,i_1,\ldots\}$, where each  $i_k$ is a number in the range $-L, -L+1, \ldots, L$.
Presence  of a number $i_k$ at some position of the code  means that at the
corresponding $T$-iteration the orbit visits the island $D_{i_k}$.
%(see \ref{def:sigma}).

\subsection{Indexing  of $h$-strips in an island}
\label{Indexing_hv}

Consider the action of the Poincare map $T$ on an island $D_i$.
The points that  are sent by $T$ from the island $D_i$ to an island $D_j$
constitute the set $T^{-1}D_{j}\cap D_{i}\subseteq\Delta^+_1$.
Due to Hypothesis 2, this set is an $h$-strip. Denote
\begin{equation}\label{eq:h_double}
  H_{ij}\equiv T^{-1}D_j\cap D_i.
\end{equation}
By analogy, consider points  that visit the islands $D_{i_{0}},
 D_{i_{1}}, \ldots,$ $D_{i_n}$ under the action of $T$  in the given order.
Any point $p$ in $D_{i_0}$  whose orbit  visits the islands in this
order belongs to an $h$-strip $H_{i_0\ldots i_n}$ which can be constructed using the
following recurrence rule
\begin{equation}\label{eq:h_multi}
  H_{i_{k}\ldots i_n}\equiv
  T^{-1}H_{i_{k+1}\ldots i_n}\cap
  D_{i_k}.
\end{equation}
%By construction, $H_{i_0i_1\ldots i_n}\subseteq\Delta_{n}^+$.
In other words, the $h$-strip $H_{i_0\ldots i_n}$ coincides with  the
set of points $p\in \mathbb{R}^2$ satisfying  three conditions:
\begin{enumerate}
  \item[(a)] $p\in D_{i_0}$;
  \item[(b)] there exist $Tp$, $T^2 p$, \ldots, $T^n p$;
  \item[(c)] $Tp \in D_{i_1}$, $T^2p \in D_{i_2}$,
    \ldots, $T^np \in D_{i_n}$.
\end{enumerate}

The points (a)--(c) imply that the $h$-strips with
different multi-indices are embedded in the following sense
\begin{displaymath}\label{eq:h_embedded}
  \ldots\subset H_{i_{0}\ldots i_n}\subset
  H_{i_{0}\ldots i_{n-1}}\subset\ldots \subset
  H_{i_{0}i_{1}i_2}\subset H_{i_{0}i_1}\subset D_{i_0}.
\end{displaymath}
If a point $p\in H_{i_{0}\ldots i_n}$ generates an orbit, then this orbit  contains  a block
\begin{displaymath}\label{eq:h_block}
  \{\ldots, \underbrace{i_{0},i_{1},\ldots, i_n},\ldots\}.
\end{displaymath}
By definition, $h$-strip $H_{i_0\ldots i_n}$ is bounded by two
$h$-curves. Points from one of them are $(n+1)\pi$-collapsing to $+\infty$. Denote this $h$-curve by $\alpha_{i_0\ldots i_n}^+$. By construction,
$T^n\alpha_{i_0\ldots i_n}^+\subset \alpha^+_{i_n}$. Points from another $h$-curve are   $(n+1)\pi$-collapsing to $-\infty$. Denote this $h$-curve by $\alpha_{i_0\ldots i_n}^-$ and note that
$T^n\alpha_{i_0\ldots i_n}^-\subset \alpha^-_{i_n}$.

%To provide a full information about the
%structure of $\Delta_0$ we also have to describe $v$-strips inside this set.

The indexing of $v$-strips is  introduced in a similar manner.    The points that  are sent by $T^{-1}$ from an island $D_i$ to an island $D_j$ constitute  a $v$-strip denoted as
\begin{equation}\label{eq:v_double}
  V_{ji}\equiv TD_j\cap D_i.
\end{equation}
By analogy, $v$-strip denoted as  $V_{i_{-n}\ldots i_0}$ is the set of points which belong to $D_{i_0}$ and under the action of $T^{-1}$ visited   the islands $D_{i_{0}}$, $D_{i_{-1}}$, $\ldots$, $D_{i_{-n}}$ in the given order. This $v$-strip can be constructed using the  recurrence rule
\begin{equation}\label{eq:v_multi}
  V_{i_{-n}\ldots i_{-k}}\equiv
  TV_{i_{-n}\ldots i_{-k-1}}\cap
  D_{i_{-k}}.
\end{equation}
%By construction, $V_{i_{-n}i_{-n+1}\ldots i_0}\subseteq\Delta_n^-$.
%In other words, the $v$-strip $V_{i_{-n}i_{-n+1}\ldots i_0}$ is the set of points
%$p\in \mathbb{R}^2$ which satisfy the three conditions:
%\begin{enumerate}
%  \item[(a)] $p\in D_{i_0}$;
%  \item[(b)] there exist $T^{-1}p$, $T^{-2} p$,
%    \ldots, $T^{-n} p$;
%  \item[(c)] $T^{-1}p \in D_{i_{-1}}$, $T^{-2}p
%    \in D_{i_{-2}}$, \ldots, $T^{-n}p \in D_{i_{-n}}$.
%\end{enumerate}
The $v$-strips with different multi-indices are embedded according to
the following rule
\begin{displaymath}
  \ldots\subset V_{i_{-n}\ldots i_0}\subset
  V_{i_{-n+1}\ldots i_0}\subset\ldots \subset
  V_{i_{-2}i_{-1}i_0}\subset V_{i_{-1}i_0}\subset D_{i_0}.
\end{displaymath}
If a point  $p\in V_{i_{-n}\ldots
i_0}$  generates an orbit, then this orbit  contains a block
\begin{displaymath}
  \{\ldots, \underbrace{i_{-n},i_{-n+1},\ldots,
  i_0},\ldots\}.
\end{displaymath}
%Then, the orbit generated by a point $p\in D_{i_0}$ that lies in
%the intersection $V_{i_{-n}i_{-n+1}\ldots i_0}\cap
%H_{i_{0}i_{1}\ldots i_n}$ has in its code a block
%\begin{displaymath}
%  \{\ldots, \underbrace{i_{-n},i_{-n+1},\ldots, i_n},\ldots\}.
%\end{displaymath}

\subsection{The fixed points of $T$ and their stable and unstable
manifolds}
\label{subStruct}

%Let us prove the following statement.
The following statement is valid.
\begin{theorem}\label{prop1}
%  Let Hypotheses 1-3 be valid. Then

        (a) Each island $D_{i^*}$ contains one and only one fixed
      point ${\mathcal{O}}_{i^*}$ of $T$, corresponding to the code
      \begin{displaymath}
        \{\ldots,i^*,i^*,i^*,\ldots\}.
      \end{displaymath}

	(b) The fixed point ${\mathcal{O}}_{i^*}$ possesses local unstable
    manifold ${\mathcal{V}}_{i^*}$  and local stable manifold   ${\mathcal{H}}_{i^*}$ of
    $T$. The local unstable manifold  ${\mathcal{V}}_{i^*}$ is a v-curve and
      \begin{align}
          \{p\in {\mathcal{V}}_{i^*}\} \iff \{T^{-n}
          p\in D_{i^*} \mbox{\, for\,  } n=0,1,\ldots, \mbox{\, and \,} \lim_{n\to\infty}T^{-n} p= {\mathcal{O}}_{i^*}\}.\label{UnstM}
      \end{align}
The local stable manifold  ${\mathcal{H}}_{i^*}$ is an h-curve and
      \begin{align}
        \{p\in {\mathcal{H}}_{i^*}\} \iff \{T^n p\in
          D_{i^*} \mbox{\, for\,  } n=0,1,\ldots, \mbox{\, and \,} \lim_{n\to\infty}T^n p= {\mathcal{O}}_{i^*}\}.\label{StM}
      \end{align}

	(c) The code $\{\ldots,i^*,i^*,i_1,i_2,\ldots\}$,  $i_1\ne
      i^*$, corresponds to an orbit $\{\ldots, p_{-1},$
      $p_0,p_1, \ldots\}$, such that~
%      \begin{displaymath}
       $ \lim_{n\to\infty} p_{-n}={\mathcal{O}}_{i^*}$,
%      \end{displaymath}
      whereas the  code $\{\ldots,i_1,i_2,i^*,i^*,\ldots\}$, $ i_2\ne
      i^*$, corresponds to an orbit $\{\ldots,$ $p_{-1},p_0,
      p_1, \ldots\}$, such that~
%      \begin{displaymath}
        $\lim_{n\to\infty} p_{n}={\mathcal{O}}_{i^*}$.
%      \end{displaymath}
%  \end{itemize}
\end{theorem}

The proof of Theorem~\ref{prop1} is postponed to Appendix~\ref{App_A}.

The origin $(0,0)$ is a fixed point of $T$, and it belongs to one of the islands. It is convenient to enumerate the islands in such a way that the origin belongs to $D_0$, i.e.,  ${\mathcal{O}}_0=(0,0)\in D_0$.
\begin{corollary}\label{prop1_cor1}
  %Let Hypotheses 1-3 and Hypothesis+ hold and ${\mathcal{O}}_0=(0,0)\in D_0$.
  %Then
  The codes
  \begin{displaymath}
    \{\ldots,0,0,i_1,i_2,\ldots,i_m,0,0\ldots\}
  \end{displaymath}
  correspond to homoclinic loops of ${\mathcal{O}}_0=(0,0)$.
\end{corollary}

\begin{remark}\label{Rem31}
  Since (\ref{1D_rep}) is invariant with respect to the change
  $\psi(x)\to -\psi(x)$, the $v$-curve ${\mathcal{V}}_{0}$ and $h$-curve
  ${\mathcal{H}}_{0}$ are symmetric with respect to the origin.
\end{remark}

\begin{remark}\label{Rem32}
  Let  $DT({\mathcal{O}}_{i^*})$ be the operator of linearization for
  $T$ at fixed point ${\mathcal{O}}_{i^*}\in D_{i^*}$. Assume that
  the eigenvalues $\lambda_1$ and $\lambda_2$ of $DT({\mathcal{O}}_{i^*})$ are
  such that $\lambda_{1,2}^n\ne 1$, $n=1,2,3,4$
  (see \cite{ArnoldMM}, Addition 8). Since $T$ is an
  area-preserving map, $\lambda_1\lambda_2=1$. Then, due to
  existence of local stable and unstable manifolds, ${\mathcal{H}}_{i^*}$ and ${\mathcal{V}}_{i^*}$, the fixed point
  ${\mathcal{O}}_{i^*}$ is of hyperbolic type. Assume that
  $|\lambda_1|>1$, $|\lambda_2|<1$. If the corresponding
  eigenvectors of $\lambda_1$ and $\lambda_2$ are denoted by
  ${\rm e}_1$ and  ${\rm e}_2$, then the tangent vector to
  ${\mathcal{V}}_{i^*}$  is ${\rm e}_1$  and the tangent
  vector to ${\mathcal{H}}_{i^*}$  is ${\rm e}_2$.
\end{remark}

%TO CHANGE, TO MOVE

\subsection{Ordering of islands in $\Delta_0$}
\label{OrdIsl}
In the previous discussion, we used subscripts $i=0,\pm 1, \ldots, \pm L$  to  label   islands  $D_i$ in $\Delta_0$. It was convenient to assume that the  island containing the origin corresponds to $i=0$, but we have  never specified the enumeration of   other islands. Let us now fix the ordering of the islands using the following concepts.

\begin{definition}\label{def+}
A curve $\gamma$ is called {\it an $\infty$-curve} if $\gamma=T\beta$ where $\beta\subset D$ is a v-curve and $D$ is an arbitrary island.
\end{definition}

%\begin{definition}\label{def23}
      % One can consider the nonlinear equation
      % with term 1/psi_x. In this case all islands are situated
      % only on psi_x axis, hence we have to redefine this property.
%    \item $\gamma\cap D_i$ is a $v$-curve for any $i=-L,\ldots,L$.
%  \end{enumerate}
%\end{definition}

\begin{lemma}\label{lemma2}
  An {\sl $\infty$-curve} is a
  plane Jordan curve without self-intersections, $\gamma=\{(\psi(\xi),
  \psi^\prime(\xi)),\,\xi\in(0;1)\}$, such that $\lim_{\xi\to0}\psi(\xi)=-\infty$, $\lim_{\xi\to1}\psi(\xi)=+\infty$.
\end{lemma}

%\begin{lemma}\label{lemma2}
 % For any $v$-curve ${\cal C}_v\subset D_i$, $i=-L,\ldots,L$, the image
%  $T{\cal C}_v$ is an $\infty$-curve.
%\end{lemma}

\begin{proof}
By definition, $\gamma=T\beta$ and $\beta$ is a $v$-curve.
The endpoints of $\beta$, $b^-$ and $b^+$,
  are situated on $\alpha^-_{i}$ and $\alpha^+_{i}$, respectively. Introduce on  $\beta$ a parametrization
  $q(\xi)=(\psi(\xi);\psi^\prime(\xi))$, $\xi\in (0;1)$  in such a way that
  $b^-=q(0)$, $b^+=q(1)$.  The image $T\beta$ is a
  continuous curve without self-intersection that by Hypothesis 2 crosses
  each of the islands  $D_j$, $j=-L,\ldots,L$, and each intersection
  $T\beta\cap D_j$ is a $v$-curve. Moreover, by Hypothesis+ the
  points $q(0)$ and $q(1)$
  are $\pi$-collapsing points, $q(0)$ collapses to $-\infty$ and $q(1)$
  collapses to $+\infty$.  This implies the statement of Lemma~\ref{lemma2}.
\end{proof}

\begin{corollary}\label{cor25}
  $T{\mathcal{V}}_i$ for any $i=-L,\ldots,L$ are $\infty$-curves. Moreover,
  ${\mathcal{V}}_i=T{\mathcal{V}}_i\cap D_i$.
\end{corollary}

\begin{definition}\label{def24}
  Let $\gamma=\{q(\xi), \xi\in(0;1)\}$  be an  $\infty$-curve. For two
  points $q_1=q(\xi_1)\in\gamma$, $q_2=q(\xi_2)\in \gamma$ we say that
  $q_2\succ q_1$ (or $q_1\prec q_2$) if $\xi_1<\xi_2$.
\end{definition}

Each $\infty$-curve passes through each island exactly once.  Let us introduce the order of island with respect to an $\infty$-curve.

\begin{definition}\label{def25}
  Let $\gamma=\{q(\xi), \xi\in(0;1)\}$  be an $\infty$-curve and  $D_A$ and $D_B$ be  islands.   Let $\gamma\cap D_A$ be  $v$-curve $\{q(\xi), \xi\in(\xi_A^-;\xi_A^+)\}$ and  $\gamma\cap D_B$ be $v$-curve  $\{q(\xi), \xi\in(\xi_B^-;\xi_B^+)\}$. We say that
$D_B\succ D_A$ (or $D_A\prec D_B$) if $\xi_A^+<\xi_B^-$.
\end{definition}

The following lemma is valid.

\begin{lemma}\label{lemma4}
	%Let  Hypotheses 1-3 and Hypothesis+ hold.
	All $\infty$-curves pass through the  islands  in $\Delta_0$ in the same order.
\end{lemma}
\begin{proof}
	We outline the proof as follows. First of all, let us observe that by simple geometric reasons, if two $\infty$-curves $\gamma_1=T\beta^{(1)}$ and $\gamma_2=T\beta^{(2)}$
	($\beta^{(1)}$ and $\beta^{(2)}$ are $v$-curves) do not intersect
	they pass the islands in the same order. This means that if $\beta^{(1)}$ and $\beta^{(2)}$ do not intersect then the statement of the lemma holds.
	If $\beta^{(1)}$ and $\beta^{(2)}$ intersect they belong to the same island $D_{i^*}$. Then in  $D_{i^*}$ there exists
	one more $v$-curve $\beta^{(3)}$ that intersects neither  $\beta^{(1)}$ nor $\beta^{(2)}$.
 %(it follows from the fact that $\Delta_0$ is an open set and $D_{i^*}$ is an arbitrary disjoined subset of $\Delta_0$).
	By transitivity, one concludes that the statement of the lemma  holds for all $\infty$-curves.
\end{proof}

Due to Lemma~\ref{lemma4}, we can use an  arbitrary $\infty$-curve to order the islands.
Consider the ordering of island with respect to $\infty$-curve $T{\mathcal{V}}_0$. Since
$T{\mathcal{V}}_0$ is symmetric with respect to the origin, the origin
$(0,0)$ belongs to the ``central'' island $D_0$,  and the islands $D_i$ and $D_{-i}$, $i=1, 2, \ldots, L$ are situated
symmetrically with respect to the origin.
Let us  enumerate the islands  according to the introduced order. So in what follows we assume that
\begin{equation}\label{OrderDNat}
D_{-L}\prec\ldots \prec D_{-1} \prec D_0
\prec D_{1}\prec\ldots\prec D_{L}.
\end{equation}

%Assume that Hypotheses 1-3 and Hypothesis+ hold.
\begin{definition}\label{def:h_i}
  Let $\gamma$ be an $\infty$-curve. Denote
  \begin{displaymath}
    h_i^\pm(\gamma)\equiv
    \gamma\cap \alpha_i^\pm.
  \end{displaymath}
\end{definition}

%\begin{remark}\label{rem:compact}
%For the sake of the compactness, in what follows if some statement holds for any $\infty$-curve $\gamma$ then we will omit the agrument $\gamma$, i.e. $h_i^\pm=h_i^\pm(\gamma)$. It also means that the point $h_i^+$ is an arbitrary point of the boundary $\alpha_i^+$ as well as $h_i^-$ is an arbitrary point of $\alpha_i^-$.
%\end{remark}

\begin{remark}
By construction, the points $h_i^\pm(\gamma)$ are $\pi$-collapsing to $\pm\infty$ points.
\end{remark}

\begin{remark}\label{lemma5}	
By simple geometric arguments, one can prove that  if  $h_{i}^-(\gamma)\prec h_{i}^+(\gamma)$ for some $\infty$-curve $\gamma$, then the same relation $h_{i}^-(\gamma_1)\prec h_{i}^+(\gamma_1)$ holds for any other  $\infty$-curve $\gamma_1$; vice versa if  $h_{i}^+(\gamma)\prec h_{i}^-(\gamma)$ for some $\infty$-curve $\gamma$, then the same relation  holds for any other  $\infty$-curve.
\end{remark}

%In informal terms, all $\infty$-curves ``enter'' and ``exit'' each particular island in the %same order.
\begin{remark}\label{lemma15}	 One can also prove that for any two successive islands $D_i$ and $D_{i+1}$ in (\ref{OrderDNat}) and for arbitrary $\infty$-curve $\gamma$ one of the two alternatives holds:
%
% the order of the points $h_i^\pm(\gamma)$ for two
% neighboring islands $D_i$ and $D_{i+1}$
  \begin{displaymath}
    h_{i}^-(\gamma)\prec h_{i}^+(\gamma)\prec
    h_{i+1}^+(\gamma)\prec h_{i+1}^-(\gamma)\quad \mbox{or}\quad
     h_{i}^+(\gamma)\prec h_{i}^-(\gamma)\prec
    h_{i+1}^-(\gamma)\prec h_{i+1}^+(\gamma).
  \end{displaymath}
\end{remark}

%\begin{proof}
%  Assume that
%  \begin{displaymath}
%    h_{i}^\mp\prec h_{i}^\pm\prec
%    h_{i+1}^\mp\prec h_{i+1}^\pm
%  \end{displaymath}
%  and consider the segment of $\gamma$ between the intermediate points
%  $h_{i}^\pm$ and $h_{i+1}^\mp$. This segment is
%  continuous, its endpoints are $\pi$-collapsing to $\pm\infty$ points.
%  At the same time other points of
%  this segment lie outside $\Delta_0$, therefore they are $P$-collapsing
%  points where $0<P<\pi$. This situation contradicts to the theorem
%  about continuous dependence of solution of \ref{1D_rep} on initial
%  data.
%\end{proof}

\begin{remark}
  In simple terms, if ${\mathcal{U}}_\pi^+$ and ${\mathcal{U}}_\pi^-$ are
  %well-pronounced
  curvilinear strips (as in Figure~\ref{U_pi}), then the
  $\infty$-curves pass ``along'' the strip $\mathcal{U}_\pi^-$,
  and the order (\ref{OrderDNat}) is the order in which the
  islands appear in ${\mathcal{U}}_\pi^-$.
\end{remark}

\begin{figure}
  \centering
  \includegraphics [width=\textwidth]{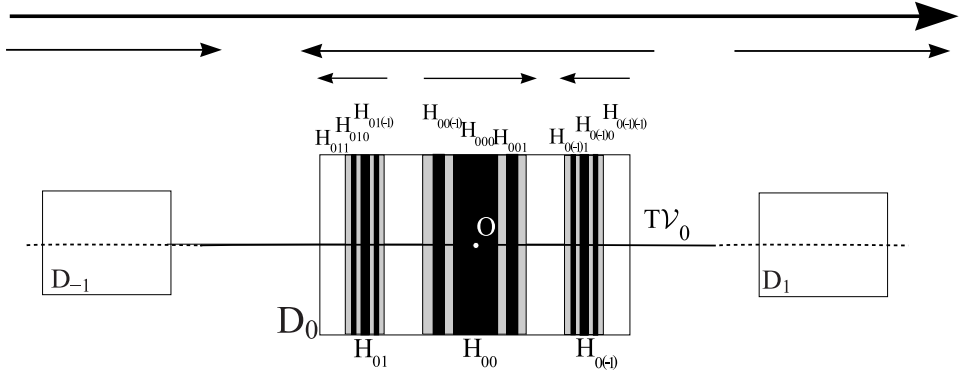}
  \caption{Ordering of $h$-strips in the island $D_0$ in the case
    $L=1$. The island $D_0$ contains three $h$-strips with multi-indices of
    length  two: $H_{0(-1)}$, $H_{00}$, $H_{01}$. Each of
    these $h$-strips contains three  $h$-strips with multi-indices of the length
    three. Within each of the $h$-strips the ordering of embedded $h$-strips inherits  the 		    pattern determined by the ``global'' arrow (in bold) and arrows over each of the islands.}
    \label{Ordering}
\end{figure}

\begin{figure}
  \centering
  \includegraphics [width=\textwidth]{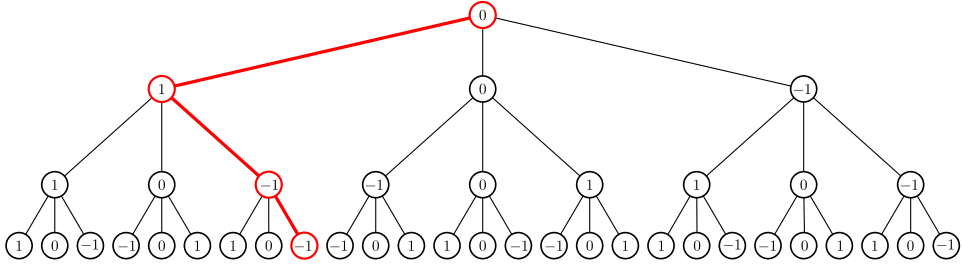}
  \caption{The same  ordering pattern as in Figure~\ref{Ordering} visualized as a ternary tree ($2L+1=3$). As an example, red path (shown by thick lines  in the grayscale version of the figure) indicates the position of $h$-strip $H_{01(-1)(-1)}$. As the graph shows, this $h$-strip is situated inside $H_{01(-1)}$, to the right from   $H_{01(-1)0}$.}
    \label{Tree}
\end{figure}

\subsection{Ordering of $h$-strips}
\label{Ordering_hv}

Consider the $h$-strip $H_{i_0\ldots i_n}\subset D_{i_0}$
defined in (\ref{eq:h_multi}).
%Notice that the boundaries $\alpha_{i_0i_1\ldots i_n}^\pm$ of
%$H_{i_0i_1\ldots i_{n}}$ defined by \ref{eq:alpha_multi} consists of
%$(n+1)\pi$-collapsing to $\pm\infty$ points.
The strip $H_{i_0\ldots i_n}$  contains $N$ non-intersecting $h$-strips $H_{i_0\ldots i_{n}i}$, where $i=-L,
\ldots,L$. Let us address the problem of ordering of $H_{i_0\ldots
i_{n}i}$ within $H_{i_0\ldots i_n}$.

\begin{definition}\label{def:h_many}
Let $\gamma$ be an $\infty$-curve. Denote
  \begin{displaymath}
    h_{i_0\ldots i_n}^\pm(\gamma)\equiv
    \gamma\cap\alpha_{i_0\ldots i_n}^\pm.
  \end{displaymath}
\end{definition}

%\begin{remark}
%Sometimes we will omit the argument $\gamma$, i.e. $h_{i_0\ldots i_n}^\pm=h_{i_0\ldots i_n}^\pm(\gamma)$, see \ref{rem:compact}.
%\end{remark}

\begin{definition}
  Let $\gamma=\{q(\xi), \xi\in(0;1)\}$  be an $\infty$-curve and  $H_A$ and $H_B$ be   $h$-strips.   Let $\gamma\cap H_A$ be a segment of curve $\{q(\xi), \xi\in(\xi_A^-;\xi_A^+)\}$ and  $\gamma\cap H_B$ be a segment of curve  $\{q(\xi), \xi\in(\xi_B^-;\xi_B^+)\}$. We say that
$H_B\succ H_A$ (or $H_A\prec H_B$) if $\xi_A^+<\xi_B^-$.
\end{definition}

%Note that $h_{i_0i_1\ldots i_n}^-(\gamma)$ is
%$(n+1)\pi$-collapsing to $-\infty$ point and  $h_{i_0i_1\ldots
%i_{n-1}i_n}^+(\gamma)$ is $(n+1)\pi$-collapsing to $+\infty$ point.
%Then the following statement is valid.

%REWRITE THIS SOMEHOW
\begin{theorem}\label{prop2}
Assume that % Hypotheses 1-3 and Hypothesis+ hold and
the ordering of islands in $\Delta_0$ is given by (\ref{OrderDNat}).
Consider the $h$-strip $H_{i_0\ldots i_n}$ and the embedded $h$-strips $H_{i_0\ldots i_ni}$, $i=-L,\ldots,L$.
Then for any $\infty$-curve $\gamma$ the following relations hold
\begin{itemize}

\item[(a)] $h_{i_0\ldots i_n}^-(\gamma) \prec h_{i_0\ldots i_n}^+(\gamma) \implies
H_{i_0\ldots i_n(-L)} \prec H_{i_0\ldots i_n(-L+1)} \prec\ldots\prec H_{i_0\ldots i_n(L)}$
\item[(b)] $h_{i_0\ldots i_n}^+(\gamma)  \prec h_{i_0\ldots i_n}^-(\gamma) \implies
H_{i_0\ldots i_n(L)} \prec H_{i_0\ldots i_n(L-1)} \prec\ldots\prec H_{i_0\ldots i_n(-L)}$

\end{itemize}
%(for notational simplicity we suppress the dependence on $\gamma$ in the formulas above, so $h_{i_0i_1\ldots i_n}^\pm\equiv h_{i_0i_1\ldots i_n}^\pm(\gamma)$).
\end{theorem}

\begin{proof}
We outline the proof (by induction) as follows. First, we establish the order of
 $h$-strips $H_{i_0(-L)}$, $H_{i_0(-L+1)}$, \ldots, $H_{i_0(L)}$ within the island $D_{i_0}$ depending on the situation (a) or (b), (the base of induction). Then, by applying the map $T$  it is straightforward to establish the ordering of
$H_{ i_0\ldots i_{n+1}(-L)}$, $H_{i_0\ldots i_{n+1}(-L+1)}$, \ldots, $H_{i_0\ldots i_{n+1}(L)}$ within $H_{ i_0\ldots i_{n+1}}$ if the order of  $H_{i_0\ldots i_n(-L)}$,\ldots  $H_{i_0\ldots i_n(L)}$ is known (the step of induction).
\end{proof}

%\begin{remark}
%  In simple terms, the order of $h$-strips of $(n+1)$th level within an
%  $h$-strip of $n$th level repeats the order of islands $D_i$, $i=-L,\ldots,
%  L$ within the set $\Delta_0$. The peculiarity is that the last index of
%  embedded $h$-strips increases in direction from  $\alpha_{i_0\ldots
%  i_{n}}^-$ to $\alpha_{i_0\ldots i_{n}}^+$.
%\end{remark}

\begin{remark}\label{prop3}
It is also straightforward to prove that if Hypotheses~\ref{hyp1}--\ref{hyp3} and Hypothesis\hyperref[hyp+]{+} hold and the ordering of islands in $\Delta_0$ is given by (\ref{OrderDNat}) then the following implications are valid
\begin{itemize}
\item[(a)]
$h_{-L}^-(\gamma)\prec
h_{-L}^+(\gamma), \,\,
h_{i_0\ldots i_n}^-(\gamma)\prec
h_{i_0\ldots i_n}^+(\gamma)\;\Rightarrow\;
h_{i_0\ldots i_n(-L)}^-(\gamma)\prec
h_{i_0\ldots i_n(-L)}^+(\gamma)$
\item[(b)]
$h_{-L}^-(\gamma)\prec
h_{-L}^+(\gamma), \,\,
h_{i_0\ldots i_n}^+(\gamma)\prec
h_{i_0\ldots i_n}^-(\gamma)\;\Rightarrow\;
h_{i_0\ldots i_n(-L)}^+(\gamma)\prec
h_{i_0\ldots i_n(-L)}^-(\gamma)$
\item[(c)]
$h_{-L}^+(\gamma)\prec
h_{-L}^-(\gamma),\,\,
h_{i_0\ldots i_n}^-(\gamma)\prec
h_{i_0\ldots i_n}^+(\gamma)\;\Rightarrow\;
h_{i_0\ldots i_n(-L)}^+(\gamma)\prec
h_{i_0\ldots i_n(-L)}^-(\gamma)$
\item[(d)]
$h_{-L}^+(\gamma)\prec
h_{-L}^-(\gamma), \,\,
h_{i_0\ldots i_n}^+(\gamma)\prec
h_{i_0\ldots i_n}^-(\gamma)\;\Rightarrow\;
h_{i_0\ldots i_n(-L)}^-(\gamma)\prec
h_{i_0\ldots i_n(-L)}^+(\gamma)$
\end{itemize}
Here $\gamma$ is an arbitrary $\infty$-curve. These implications establish
the relation between the orientation of boundaries of the leftmost island $D_{-L}$, the strip $H_{i_0\ldots i_n}$ and the ordering of the strips $H_{i_0\ldots i_n i}$, $i=-L,\ldots,L$ within $H_{i_0\ldots i_n}$.
\end{remark}

{
If the order of the islands and the positions of their boundaries $\alpha_{i}^{\pm}$ are known (with respect to an arbitrary $\infty$-curve), then the established relations are sufficient to describe   positions of $h$-strips with different indexes. The ordering scheme   can be sketched as follows. It is necessary to draw the islands $D_{-L}$,  $D_{-L+1}$, $\ldots$, $D_L$ from the left to the right,  to draw the ``global'' arrow in the same direction (the upper bold arrow in Figure~\ref{Ordering}). Next, it is necessary  to draw an arrow over each island such that any of those arrows is directed from $\alpha_{i}^{-}$ to   $\alpha_{i}^{+}$. According to Remark~\ref{lemma15}, the directions of any two successive arrows will be opposite. Then  one can continue the procedure recursively, drawing arrows of higher levels and identifying the position of $h$-strips with different indexes following to directions of the arrows.    In simple terms,   the order of $h$-strips of $(n+1)$th level within an
  $h$-strip of $n$th level repeats the order of islands $D_i$, $i=-L,\ldots,
  L$ within the set $\Delta_0$. The peculiarity is that the last index of
  embedded $h$-strips increases in direction from  $\alpha_{i_0\ldots
  i_{n}}^-$ to $\alpha_{i_0\ldots i_{n}}^+$. An example of the ordering procedure is presented in Figure~\ref{Ordering}.  As follows from the ordering procedure, the positions of $h$-strips with different indexes can be catalogued conveniently using a  $(2L+1)$-ary tree graph, where the root of the graph corresponds to the island $D_0$ and leaves of increasing levels denote   $h$-strips with indexes of growing length. Figure~\ref{Tree}   presents an example of such a tree   for $L=1$ (four levels of the tree are shown).}

\section{Algorithm for constructing of a gap soliton by the given code}
\label{subAlgorithm}
%In what follows we assume that Hypotheses 1-3 and Hypothesis+ hold.

\subsection{General remarks}
\label{GenRem}

Let $\widehat\psi(x)$ be a gap soliton solution of (\ref{1D_rep}). Then
$\widehat\psi(x)$ can be associated with an orbit  $\textbf{s} = \{\ldots, p_{-1},
p_0, p_1, \ldots\}$ whose entries are  defined by  $p_n = (\widehat\psi(n\pi),
\widehat\psi_x(n\pi))$, $n\in\mathbb{Z}$.  Since $\widehat\psi(x)$ satisfies the localization condition (\ref{BoundCond}), the orbit $\textbf{s}$
%corresponding to the gap soliton $\widehat\psi(x)$
is a homoclinic loop of the fixed point ${\mathcal{O}}_0=(0,0)$, i.e.,
$ \lim_{n\to\pm\infty} p_n={\mathcal{O}}_0$.
By the coding theorem, the orbit $\textbf{s}$ is  associated with  an unique code $a\in \Omega^N$ which has only a finite number of nonzero entries,
i.e., it has the form
\begin{equation}\label{GapCode}
  a=\{\ldots,0,0,i_1,i_2,\ldots,i_m,0,0,\ldots\},
  \quad i_1\ne 0,\quad i_m\ne 0.
\end{equation}
This   code is naturally associated with the gap soliton  $\widehat\psi(x)$.
By definition of $\Sigma$,  the code (\ref{GapCode}) determines
the order in which  orbit  $\textbf{s}$ visits the islands $D_i$, $i=-L,\ldots,L$.
Therefore this code can be obtained having the location of the islands
and $m$ coordinates of the points $(\widehat\psi(n\pi),\widehat\psi_x(n\pi))$, $n=1,\ldots,m$.
%Notice that the points $(\widehat\psi(\pi),\widehat\psi_x(\pi))$ and $(\widehat\psi(m\pi),\widehat\psi_x(m\pi))$ are not situated in $D_0$.
This yields the answer to the Question~1 in Section~\ref{Intro}.

Conversely, consider a code $a\in\Omega^N$ of the form
(\ref{GapCode}). Due to Theorem~\ref{theorem1}, there
exists an unique orbit  $\textbf{s}$ which corresponds to this code, and
according to Corollary~\ref{prop1_cor1} this orbit is a homoclinic loop of the fixed point ${\mathcal{O}}_0=(0,0)$. Let
$p=(\psi_0, \psi'_0)$ be  an arbitrary point of this orbit. Consider
$\widehat\psi(x)$, the solution of the Cauchy problem for (\ref{1D_rep})
with the initial conditions
\begin{equation}\label{IC}
  \widehat\psi(0)=\psi_0,  \quad \mbox{ and}
  \quad \widehat\psi_x(0)=\psi'_0.
\end{equation}
This solution is a gap soliton, i.e., $\widehat\psi(x)\to 0$ as $x\to\pm\infty$.
%The solution $\psi(x)$ is determined by the code \ref{GapCode} up to a shift with
%respect to $x$: choosing  another entry of $\textbf{s}$ as initial data
%for the Cauchy problem we results in the solution $\psi(x+n\pi)$
%where $n$ is a nonzero integer.  Anyway, $\psi(x)\to 0$ as $x\to\pm\infty$, i.e., $\psi(x)$ is a gap soliton  solution of \ref{1D_rep}.
Hence for construction of the profile of a gap soliton
$\widehat\psi(x)$ by the code (\ref{GapCode}) one
has to find {\it any} entry of the orbit $\textbf{s}$ with sufficiently good accuracy.
Then the profile of $\widehat\psi(x)$  can be recovered from the numerical
solution of the Cauchy problem (\ref{IC}).

\subsection{Description of the algorithm}\label{BrDescr}

Consider a  code  of the form (\ref{GapCode}).
Having infinite number of zeros from the left, this code corresponds to an
orbit that remains in $D_0$ for infinitely many iterations of $T^{-1}$.
The relation between the entries of the code (\ref{GapCode}), the entries of
the orbit $\textbf{s}=\{\ldots p_{-1},p_0,p_1,\ldots\}$, $p_n=
(\psi_n,\psi'_n)$,  and the values of the solution $\widehat\psi(x)$ of
(\ref{1D_rep}) is given in Table~\ref{Correspond}. Note
that for $n\leq0$ the points $p_n$ are situated on ${\mathcal{V}}_0$, the local
unstable manifold of ${\mathcal{O}}_0$, and for $n>m$ the points $p_n$ lie on
${\mathcal{H}}_0$, the local stable manifold of ${\mathcal{O}}_0$.

Introduce the notation
\begin{displaymath}
(0\times M)=\underbrace{00\ldots0}_{M\mbox{ times}}.
\end{displaymath}
For instance, the $h$-strip indexed by $M$ consecutive zeros and symbols
$\{i_1,i_2,\ldots, i_m\}$ is denoted by $H_{(0\times M)i_1\ldots i_m}$.
Since it is sufficient to find any entry of $\textbf{s}$, we fix $M>0$ (large enough) and seek for the entry $p_{-M}$. This entry can be characterized as follows:
%\begin{itemize}

(a) $p_{-M}\in  H_{(0\times M)i_1\ldots i_m(0\times K)}$, where $K>0$ is arbitrarily large, and
\begin{equation}
 H_{(0\times M)i_1\ldots i_m(0\times K)}\subset\ldots H_{(0\times M)i_1\ldots i_m}\subset \ldots H_{(0\times M)i_1}\subset H_{(0\times M)}\subset D_0;\label{eq:chain}
 \end{equation}

(b)  $p_{-M}\in{\mathcal{V}}_0$;

(c) $T^{M+m+1}p_{-M}\in {\mathcal{H}}_0$.
%\end{itemize}

\begin{table}
  \caption{The relation between the entries of the code (\ref{GapCode}),
    the entries of the orbit $\textbf{s}=\{\ldots p_{-1},p_0,p_1,\ldots\}$
    and the values of $\widehat\psi(x)$ of (\ref{1D_rep}),
    $i_1\ne 0$, $i_m\ne 0$}
  \begin{tabular}{|l|ccccccccc|}
    \hline
    Code $a$  &  $\ldots$  &  $0$  &  $\ldots$  &  $0$  &
      $i_1$  &  $\ldots$  &  $i_m$  &  $0$  &  $\ldots$\\
    Orbit $\bf s$  &  $\ldots$  &  $p_{-M}$  &  $\ldots$  &  $p_0$  &
      $p_1$  &  $\ldots$  &  $p_m$  &  $p_{m+1}$  &  $\ldots$\\
    Island  &  $\ldots$  &  $D_0$  &  $\ldots$  &  $D_0$  &
      $D_{i_1}$  &  $\ldots$  &  $D_{i_m}$  &  $D_0$  &  $\ldots$\\
    $\widehat\psi(x)$  &  $\ldots$  &  $\widehat{\psi}(-M\pi)$  &  $\ldots$  &  $\widehat{\psi}(0)$  &
      $\widehat{\psi}(\pi)$  &  \ldots  &  $\widehat{\psi}(m\pi)$  &  $\widehat{\psi}((m+1)\pi)$  &  $\ldots$\\
  \hline
  \end{tabular}
  \label{Correspond}
\end{table}
Let us take a closer look at the island $D_0$. It is limited by a pair of $h$-curves
$\alpha^+_0$ and $\alpha^-_0$ and a pair
of $v$-curves $\beta^+_0$ and $\beta^-_0$ (see Figure~\ref{Island(def)}). The
curves $\alpha^\pm_0$ are such that if $q=(\psi,\psi')\in\alpha^\pm_0$,
then the solution $\psi(x)$ of (\ref{1D_rep}) with these initial data at $x=0$
satisfies the condition $ \lim_{x\to\pi}\psi(x)=\pm\infty$,
(the sign $+$ or $-$ corresponds to the superscripts in
$\alpha^{\pm}_0$). According to Corollary~\ref{cor25}, $T{\mathcal{V}}_0$ is an $\infty$-curve that intersects $\alpha^+_0$ and $\alpha^-_0$, and ${\mathcal{V}}_0=T{\mathcal{V}}_0\cap D_0$. According to Definition~\ref{def:h_i}, denote
\begin{displaymath}
  h_0^-\equiv h_0^-(T{\mathcal{V}}_0)\equiv T{\mathcal{V}}_0\cap \alpha^-_0,\quad
  h_0^+\equiv h_0^+(T{\mathcal{V}}_0)\equiv T{\mathcal{V}}_0\cap \alpha^+_0.
\end{displaymath}
%The points $h_0^\pm$ correspond to initial data at $x=0$ for
%\ref{1D_rep} such that the solution $\psi(x)$ of
%the Cauchy problem satisfies the conditions
%\begin{displaymath}
%  \lim_{x\to-\infty}\psi(x)=0,\quad
%  \lim_{x\to\pi}\psi(x)=\pm\infty.
%\end{displaymath}
%Since  $\textbf{s}$ is a homoclinic orbit of ${\cal
%O}_0$, then $p_{-n}\in {\mathcal{V}}_0$,  $n\geq0$.
The points $h_0^-$ and  $h_0^+$ are rough bounds for the location
of $p_{-M}$.
Therefore $p_{-M}$
is situated  on $T{\mathcal{V}}_0$ between the points $h_0^-$ and $h_0^+$,
see Figure~\ref{Algorithm}.

%Note that
%\begin{displaymath}
% H_{0\times M,i_1i_2\ldots i_m,0\times K}\subset\ldots H_{0\times %M,i_1i_2\ldots i_m}\subset H_{0\times M}\subset D_0
% \end{displaymath}
Since $p_{-M}\in H_{0\times M}$, consider $h$-strip $H_{0\times M}$ in  detail. This $h$-strip is bounded by two $h$-curves, $\alpha^-_{(0\times M)}$ and $\alpha^+_{(0\times M)}$, and each of them intersects $T{\mathcal{V}}_0$ in
an unique point. Following Definition~\ref{def:h_many} denote these points by $h^-_{(0\times M)}$ and $h^+_{(0\times M)}$.
The following relations hold
\begin{displaymath}
h^-_{0\times M}=T^{-(M-1)}h_0^-, \quad
h^+_{0\times M}=T^{-(M-1)}h_0^+
\end{displaymath}
The point $p_{-M}$ is situated on $T{\mathcal{V}}_0$ between
the points $h^-_{0\times M}$ and $h^+_{0\times M}$.

Since the first nonzero symbol in the code (\ref{GapCode}) from the left is $i_1$,
the point $p_{-M}$  is situated in $h$-strip
$H_{(0\times M)i_1}\subset H_{0\times M}$.
The strip $H_{(0\times M)i_1}$ is
bounded by two $h$-curves, $\alpha^-_{(0\times M)i_1}$ and $\alpha^+_{(0\times M)i_1}$, and each of them intersects $T{\mathcal{V}}_0$ in
an unique point. Denote these points by $h^-_{(0\times M)i_1}$ and $h^+_{(0\times M)i_1}$. The point $p_{-M}$ is situated on $T{\mathcal{V}}_0$ between
the points $h^-_{(0\times M)i_1}$ and $h^+_{(0\times M)i_1}$.

Repeating this procedure following the chain (\ref{eq:chain}), one concludes that $p_{-M}$ is
situated on $T{\mathcal{V}}_0$ between the points $h^-_{(0\times M)i_1\ldots i_m(0\times K)}$ and $h^+_{(0\times M)i_1\ldots i_m(0\times K)}$. %Here $K$ is a positive integer that may be chosen arbitrarily large.
 These points
are the points of intersection of the boundary of $H_{(0\times
M)i_1\ldots i_m(0\times K)}$ with $T{\mathcal{V}}_0$.

Finally, since the code (\ref{GapCode}) has   infinite number of
consecutive  zeros from the right, $T^{M+m+K+1}p_{-M}\in
{\mathcal{H}}_0$.  Application of this condition completes the procedure.

\begin{figure}
  \includegraphics[width=\textwidth]{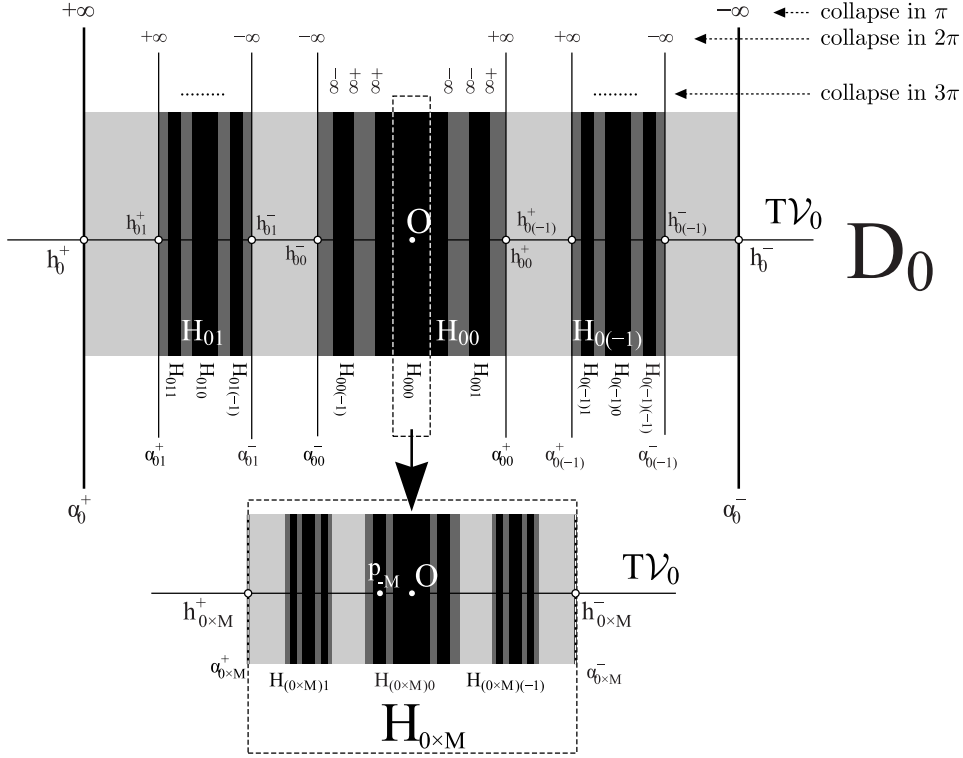}
  \caption{The scheme of the island $D_0$, positions of $h$-strips
    and the points $h_{0}^\pm$, $h_{0i_1}^\pm$ for the  case of coding using the alphabet of  three symbols (the symbols are ``$-1$'', ``$0$'' and ``$1$'').
    The point $p_{-M}$ is situated in the strip $H_{0\times M}$ between the points
    $h_{0\times M}^-$ and $h_{0\times M}^+$.
    %Its accurate
    %position depends on the code $\{i_1,i_2,\ldots,i_m\}$. If $i_1=-1$ one has to find
    %the point $p_{-M}$ in the strip $H_{0\times M(-1)}$. If $i_1=0$ then the point $p_{-M}$ must
    %be situated in the strip $H_{(0\times M)0}$. If $i_1=1$ then the point $p_{-M}$ must be situated in the strip
    %$H_{(0\times M)1}$. And so on.
    %The
    %point $p_{-M}$ is situated in the strip $H_{(0\times M)11}$, therefore
    %its localization involves finding of the points $h_{0\times M}^\pm$,
    %$h_{(0\times M)1}^\pm$ and $h_{(0\times M)11}^\pm$}
    }
  \label{Algorithm}
\end{figure}

\subsection{Implementation}
\label{Implem}

Let $\lambda_1$ be the eigenvalue of $DT({\mathcal{O}}_0)$ such
that $|\lambda_1|>1$  and its corresponding normalized
eigenvector be ${\rm e}_1=(\xi,\eta)$. Another eigenvalue of
$DT({\mathcal{O}}_0)$ is $\lambda_2=1/\lambda_1$, with
corresponding normalized eigenvector ${\rm e}_2=(\xi,-\eta)$.
The vectors ${\rm e}_{1}$ and ${\rm e}_{2}$ are tangent to
the local unstable ${\mathcal{V}}_0$ and local stable  ${\mathcal{H}}_0$ manifolds
in the point ${\mathcal{O}}_0$ correspondingly, see Remark~\ref{Rem32}. This means that if $\textbf s$ is a homoclinic
orbit of ${\mathcal{O}}_0$ with components $p_n=(\psi_n,\psi'_n)$
then
\begin{displaymath}
  \lim_{n\to+\infty}({\psi_n}/{\psi'_n})=
    -{\xi}/\eta,\quad%\label{+infty}\\
  \lim_{n\to-\infty}({\psi_n}/{\psi'_n})=
   {\xi}/\eta.%\label{-infty}
\end{displaymath}

Let us introduce the parametrization on  $T{\mathcal{V}}_0$ as
follows:  for $p\in T{\mathcal{V}}_0$ the value of parameter $\varepsilon=
\varepsilon(p)$ is  the length of the curve $T{\mathcal{V}}_0$ from the
point ${\mathcal{O}}_0$ to the point $p$ taken with the sign ``+'' if $p$
lies between ${\mathcal{O}}_0$ and $h^+_0$ and the sign ``$-$'' if $p$
is situated between ${\mathcal{O}}_0$ and $h^-_0$.  The point $p_{-M}$
is situated on the segment of $T{\mathcal{V}}_0$ between the points
$h^-_{0\times M}$ and $h^+_{0\times M}$.  For $M$ large enough
this segment  is small and can be approximated by the tangent to ${\mathcal{V}}_0$ in
${\mathcal{O}}_0$. Therefore, for homoclinic orbit of   ${\mathcal{O}}_0$ with the code (\ref{GapCode})
\begin{displaymath}
  \psi_{-M}\approx\varepsilon \xi,\quad
  \psi'_{-M}\approx\varepsilon\eta,\quad
  |\varepsilon|\ll 1.
\end{displaymath}
By the same reason, for $K$ integer and large enough the following
relation holds
\begin{displaymath}
  {\psi_{m+K}}/{ \psi'_{m+K}}
  \approx-{\xi}/{\eta}.
\end{displaymath}

Once the  parametrization is introduced, the numerical algorithm (see Algorithm~\ref{algo2}) follows the steps of subsection~\ref{BrDescr}. The main part  of Algorithm~\ref{algo2} consists in consecutive finding of the bounds  $\varepsilon(h^-_{(0\times M)i_1})$ and $\varepsilon(h^+_{(0\times M)i_1})$,  $\varepsilon(h^-_{(0\times M)i_1i_2})$ and  $\varepsilon(h^+_{(0\times M)i_1i_2})$, and so on. The algorithm exits  with the value $\varepsilon^*=\varepsilon(p_{-M})$ found with a certain (controlled) accuracy. Once $\varepsilon^*$ is found, then  the shape of the  soliton   $\widehat{\psi}(x)$ on the interval $x\in(-M\pi;K\pi)$  is described  by the solution of Cauchy problem for (\ref{1D_rep}) with initial data
$\psi(-M\pi)=\varepsilon^*\xi$, $\psi_x(-M\pi)=
\varepsilon^*\eta$. For  $x<-M\pi$ and  $x>K\pi$ the profile can be
continued with linear asymptotics, i.e., by the solutions of
(\ref{eq:linear}) vanishing at $-\infty$
and $+\infty$, respectively, which are properly scaled to ensure the continuity of
the whole solution and its derivative.

\subsection{Some details}

Algorithm~\ref{algo2} was implemented in \texttt{Python} using   libraries
\texttt{numpy} and \texttt{scipy}. It includes nested recursive procedure
for processing of each index of the gap soliton code. The following details
should be mentioned:

1. The Cauchy problem for (\ref{1D_rep}) was solved using
    standard Runge-Kutta method of 4th order with constant step.

2. We found  that the values $M=2$ or $3$ and $K=2$ or $3$ are enough for
    the most of the calculations. The accuracy was estimated by changing
    these parameters and comparing the results. In vicinity of the upper edge of
    the gaps (when the solitons are poorly localized) $M$ and $K$ should be increased.

3. The condition $\lim_{x\to n\pi}\psi(x)=\pm\infty$
    that appears in the algorithm,  was replaced by the condition
    $\psi(n\pi)=\psi_\infty$ where $\psi_\infty$ is a large positive (or,
    correspondingly, great negative) number. Most  of the calculations
    were fulfilled for $|\psi_\infty|=10^6$ and checked for
    $|\psi_\infty|=10^8$ to confirm that this choice of   $\psi_\infty$ does not affect the results.

4. The main difficulty in practical implementation of Algorithm~\ref{algo2}  is
    caused by the fact that widths of $h$-strips $H_{(0\times M)i_1i_2\ldots}$
    decrease rapidly as the number of indices grows. Therefore the difference
    \begin{displaymath}
      \left|\varepsilon \left( h_{(0\times M)i_1i_2\ldots}^+ \right) -
      \varepsilon \left( h_{(0\times M)i_1i_2\ldots}^- \right) \right|
    \end{displaymath}
    eventually becomes beyond the computer accuracy. This hinders the calculation of  solitons with large $m$ if the program  code operates with
    floating-point numbers  with short mantissa. In order to overcome this difficulty,
    the \texttt{Python}  library {\texttt{gmpy2}} for arbitrary-precision arithmetic
    was used. In principle, this allows to compute   gap solitons with codes with
    arbitrary  $m$ (the algorithm was tested for codes with $m\leq 10$).

5. The algorithm does not allow to find gap solitons in the situation when
   Hypotheses~\ref{hyp1}--\ref{hyp3} do not hold. In particular, it cannot be applied to
    compute small-amplitude gap solitons in cosine potential (\ref{CosPot}) with
    values $(\omega,A)$ situated close to the lower gap edges [see light-gray areas in (\ref{CodZones})]. In this situation,
    a numerical continuation procedure in $\omega$ or in $A$  can be applied.

\section{Results}\label{Results}
%In order to introduce the main outcomes of our work, we compute the simpliest gap %solitons from the first and the second band gaps using \ref{algo2} and present it in %what follows.
The algorithm was applied to (\ref{1D_rep}) with cosine potential
(\ref{CosPot}). Recall that according to Conjecture made in subsection~\ref{GeomCon},  in
Regions~1 and 2 on Figure~\ref{CodZones} each gap soliton $\widehat \psi(x)$ has an
unique code of the form (\ref{GapCode}).
For Region~1  the coding alphabet consists of 3 symbols,  and for Region~2 it consists of 5
symbols. We denote the symbols of the alphabet by the integer numbers $-L,\ldots,L$, where $L=1$
for Region~1 and $L=2$ for Region~2.
This means that in Region~1 the coding alphabet is $\{-1,0,1\}$ and in Region~2 the coding alphabet is $\{-2,-1,0,1,2\}$.

Let (\ref{GapCode}) be a code of a gap soliton. Let us call $m$ the \emph{length} of the code. In Region 1, the first band gap, there are 8 distinct gap solitons with codes of the length  $m\leq 3$ (all other gap solitons with $m\leq 3$  can be obtained from them by  involutions $x\to -x$ and $\psi\to -\psi$). Figure~\ref{fig:solitons1} presents the profiles of these gap solitons computed using Algorithm~\ref{algo2}. The gap soliton in Figure~\ref{fig:solitons1}(a) corresponds to the code $\{\ldots,0,1,0,\ldots\}$. This solution   is a fundamental gap soliton of the first gap
(FGS$_1$)  \cite{Zhang09_1,Zhang09_2} (in  \cite{Yang10} this solution is called  ``on-site'' gap
soliton). FGS$_1$ is the unique elementary entity for Region~1.   The code $\{\ldots,0,-1,0,\ldots\}$ corresponds to FGS$_1$ taken with minus sign (not shown in Figure~\ref{fig:solitons1}).  More complex gap solitons can be regarded as bound states of several FGS$_1$ taken with the   plus or minus sign. For instance, the ``off-site'' gap soliton in the terminology
of \cite{Yang10} has the code $\{\ldots,0,1,-1,0,\ldots\}$ and can be regarded as a
Figure~\ref{fig:solitons1}(e).

In Region~2, inside the second band gap, there are 8 distinct (up to the involutions) gap solitons
with codes of the length  $m\leq 2$. Their profiles are shown  in
Figure~\ref{fig:solitons2}.  FGS$_1$ also presents as an elementary entity in Region~2
having
the code $\{\ldots,0,2,0,\ldots\}$ [see Figure~\ref{fig:solitons2}(a)]. However there also exists another fundamental gap soliton
FGS$_2$ with the code $\{\ldots,0,-1,0,\ldots\}$ [see Figure~\ref{fig:solitons2}(e)]. This is the so-called \textit{subfundamental} soliton which was introduced in \cite{Malomed06}; solution of this type was also discussed in \cite{Zhang09_2} where it was termed to as the second fundamental gap soliton.
The code $\{\ldots,0,1,0,\ldots\}$ corresponds to
FGS$_2$ taken with minus sign.
Again, more complex gap solitons can be regarded as bound states of FGS$_1$ and FGS$_2$ taken with different signs.

While the solitons illustrated in Figure~\ref{fig:solitons1} and Figure~\ref{fig:solitons2} correspond to relatively simple solutions with short codes ($m=1$, $2$ or $3$), the algorithm is also applicable to more complex solitons and to solutions from higher gaps. For instance, in Figure~\ref{fig:solitons3} complex gap solitons, ($m=10$) from the 2-nd and the 3-rd gaps are depicted.

\begin{figure}
  \centering
  \includegraphics[width=1\textwidth]{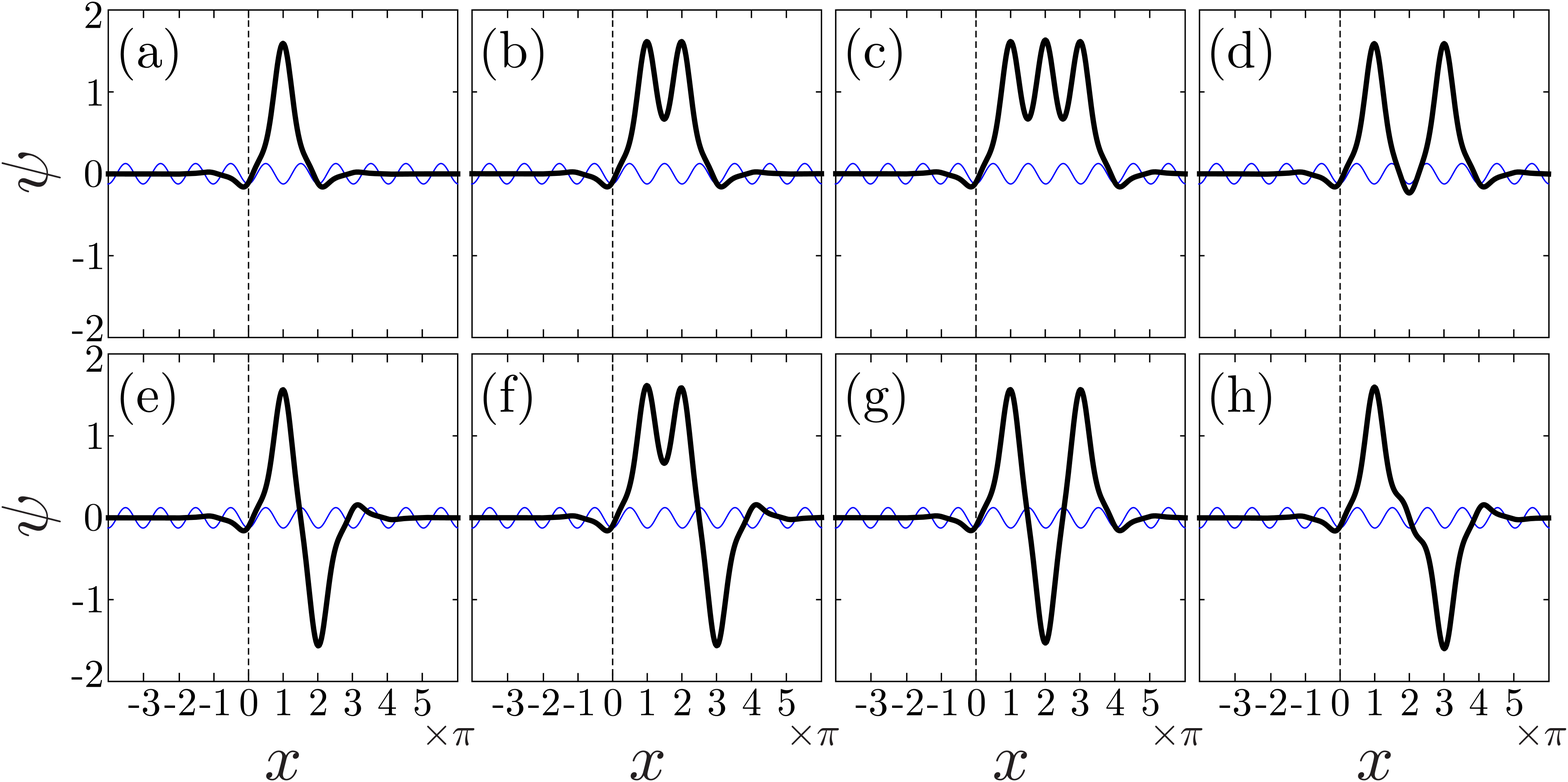}
  \caption{The gap solitons from Region 1, the first gap. Each panel shows the
    spatial profile of the soliton. The corresponding codes are: (a) $\{\ldots,0,1,0,\ldots\}$; (b) $\{\ldots,0,1,1,0,\ldots\}$; (c) $\{\ldots,0,1,1,1,0,\ldots\}$; (d) $\{\ldots,0,1,0,1,0,\ldots\}$; (e) $\{\ldots,0,1,-1,0,\ldots\}$; (f) $\{\ldots,0,1,1,-1,0,\ldots\}$; (g) $\{\ldots,0,1,-1,1,0,\ldots\}$; (h) $\{\ldots,0,1,0,-1,0,\ldots\}$. For all shown solutions   $\omega = 1$ and  $A = -3$. The vertical dashed line in each panel indicates the position of the point $p_0$, such that $Tp_0\notin D_0$, $T^{-n} p_0\in D_0$, $n\geq 0$, see the explanation in subsection~\ref{BrDescr}. Blue thin lines show  schematically the cosine potential (\ref{CosPot}).}
  \label{fig:solitons1}
\end{figure}

\begin{figure}
  \centering
  \includegraphics[width=1\textwidth]{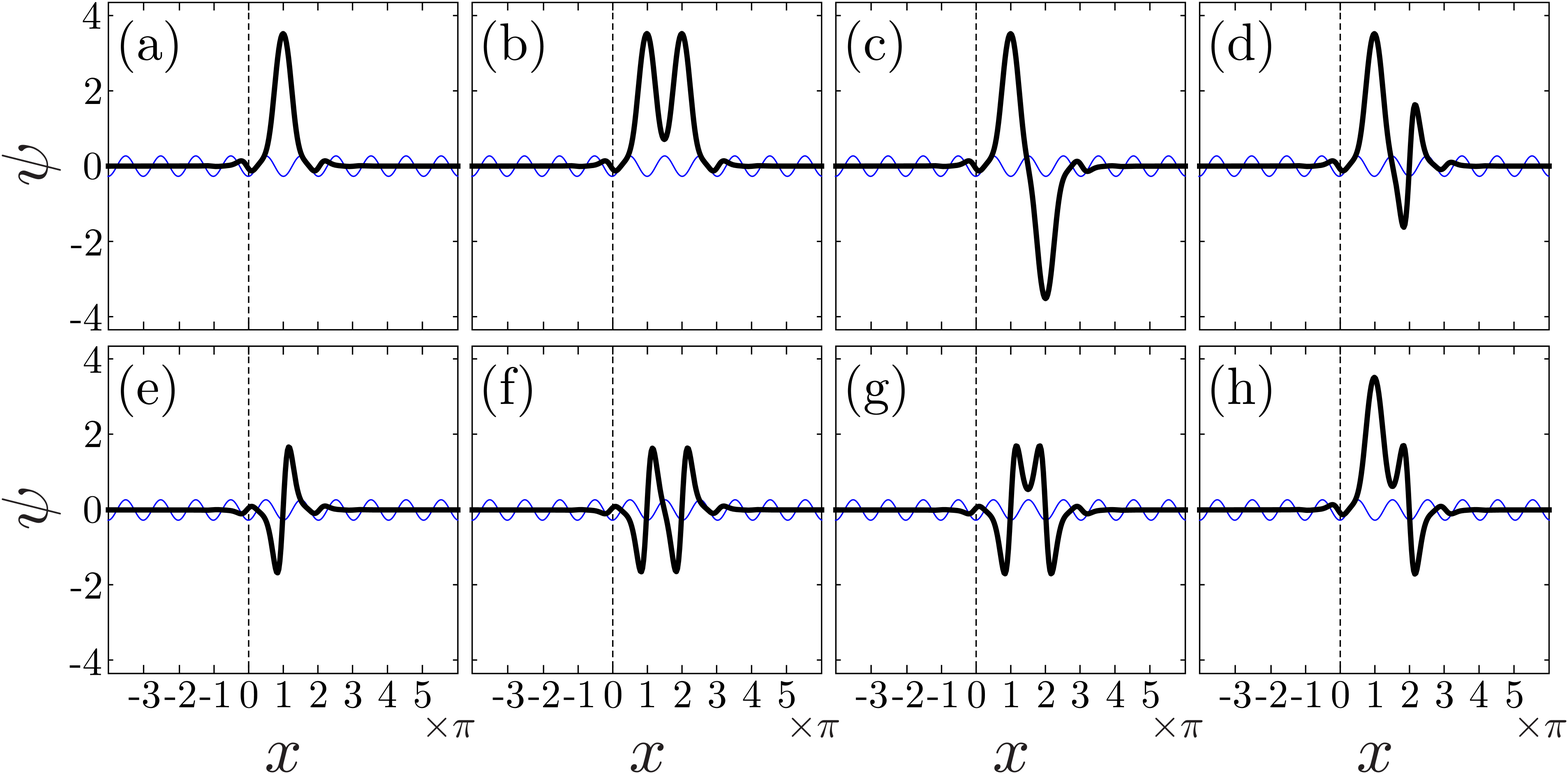}
  \caption{The gap solitons from Region 2, the second gap. The corresponding codes are: (a) $\{\ldots,0,2,0,\ldots\}$; (b) $\{\ldots,0,2,2,0,\ldots\}$; (c) $\{\ldots,0,2,-2,0,\ldots\}$; (d) $\{\ldots,0,2,-1,0,\ldots\}$; (e) $\{\ldots,0,-1,0,\ldots\}$; (f) $\{\ldots,0,-1,-1,0,\ldots\}$; (g) $\{\ldots,0,-1,1,0,\ldots\}$; (h) $\{\ldots,0,2,1,0,\ldots\}$. For all shown solutions  $\omega = 4$, $A = -10$.  The vertical line indicates the position of the point $p_0$ (see caption of Figure~\ref{fig:solitons1}), and blue thin lines show  schematically the cosine potential (\ref{CosPot}).}
  \label{fig:solitons2}
\end{figure}

\begin{figure}
  \centering
  \includegraphics[width=0.99\textwidth]{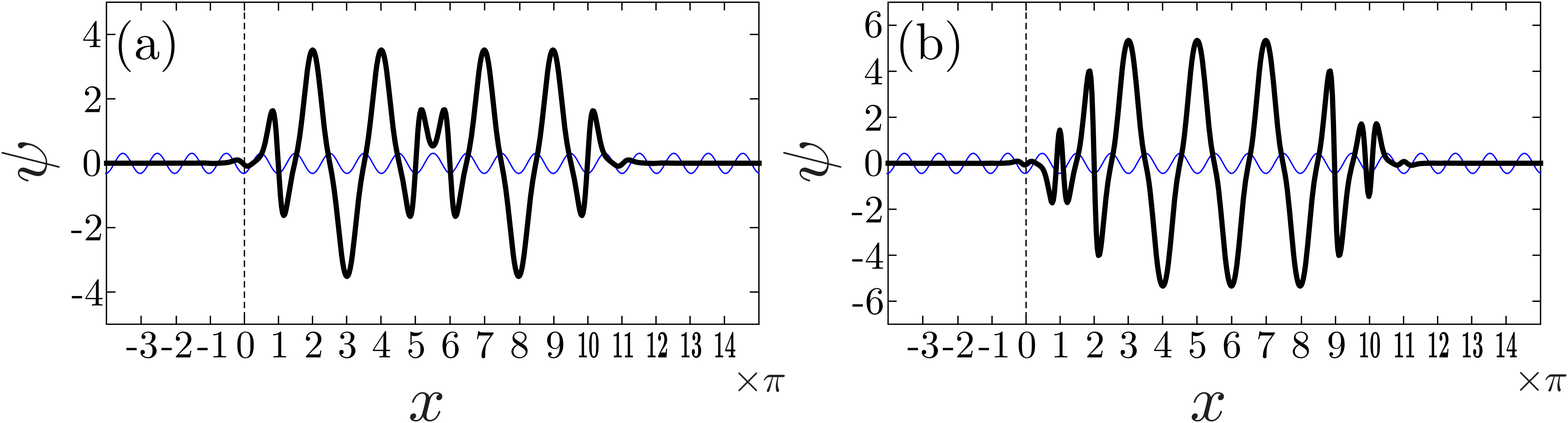}
  \caption{(a) Soliton from the second gap ($\omega=4$, $A=-10$) with the code  $\{\ldots,0, 1, 2, -2, 2, -1, 1, 2, -2, 2, -1, 0,\ldots\}$. (b)  Soliton from the third gap ($\omega=10$, $A=-20$) with the code  $\{\ldots, 0, -1, 2, 3, -3, 3, -3, 3, -3, 2, 1, 0,\ldots\}$. For each    soliton  the length  of the code  is $m=10$.}
  \label{fig:solitons3}
\end{figure}

\section{Summary and discussion}
\label{Disc}

In this paper we have proposed a numerical algorithm~\ref{algo2}, for
computation of localized modes for the spatially one-dimensional Gross-Pitaevskii
equation with a periodic potential $U(x)$. These modes, i.e., gap
solitons, are of the form $\Psi(x,t)=e^{-i\omega t}\psi(x)$ where
$\psi(x)$ satisfies (\ref{1D_rep}) and localization conditions
(\ref{BoundCond}). It is assumed that the parameters of the potential
$U(x)$ allow for coding of all these modes by bi-infinite sequences of
symbols from some alphabet of finite length. It is known \cite{AA13}, that
this coding is indeed possible, if certain assumptions (Hypotheses~\ref{hyp1}--\ref{hyp3}) hold,
see Theorem~\ref{theorem1}. In particular, according to \cite{AA13}, the
coding is possible for the potential $U(x)=A\cos 2x$, if the parameters
$\omega$ (the frequency of the mode) and $A$ (the amplitude of the
potential) belong to vast areas on the  plane $(\omega,A)$, see
Figure~\ref{CodZones}.

The coding approach  exploits the fact that the ``most'' of
solutions of (\ref{1D_rep}) tend to infinity at some finite point of real
axis. As a result the initial data for Cauchy problem for
(\ref{1D_rep}) corresponding to  bounded in $\mathbb{R}$
solutions form a fractal subset of ${\mathbb R}^2$. The initial data
corresponding to gap soliton solutions are situated in the intersection of this
fractal set with the local unstable manifold
${\mathcal{V}}_0$ of the zero fixed point ${\mathcal{O}}_0=(0,0)$. The numerical procedure
for computation of gap solitons consists in
iterative localization of the segments on ${\mathcal{V}}_0$ where the initial data
corresponding to a given code are situated.
Then the  profile of the nonlinear mode can be computed using the Runge-Kutta
method. Algorithm~\ref{algo2} was implemented in \texttt{Python} and applied
for  gap solitons in cosine potential (\ref{CosPot}).

Algorithm~\ref{algo2} allows for various generalizations. With minor changes the algorithm can be applied in the case of more complex odd nonlinearities, both autonomous  and
non-autonomous. It is known that the structure of $\mathcal{U}_\pi^\pm$ in (\ref{1D_rep}) in cubic-quintic \cite{UMJ16} and periodic \cite{UMJ15} nonlinearities are similar to that described here, and therefore the application of Algorithm~\ref{algo2} in these cases is straightforward.
For these problems, Algorithm~\ref{algo2} may be a useful tool to study  bifurcations
of gap solitons for (\ref{1D_rep}) with different codes, similar to the one fulfilled in \cite{ABK04} for the intrinsic localized modes of the DNLS .  %Keeping in mind the close relation
%between discrete and periodic models \cite{FS14,AKKS02},   the simplest
%bifurcations of solitons in the first gap probably correspond to the bifurcations of   their counterparts in DNLS. The study of bifurcation of
%more complex gap solitons is an interesting issue for further research.
%We also would like to stress that with minor changes the algorithm can be applied in the case of more complex odd nonlinearities, %both autonomous \cite{UMJ16} and
%non-autonomous \cite{UMJ15}.

Another possible application of Algorithm~\ref{algo2} is the study of stability of gap solitons. An interesting question is whether it is possible to predict stability or instability of a gap soliton by its code. The results of this study will be reported elsewhere.

%For acknowledgements section, please don't number the section, please begin it with \section*{Acknowledgements}
%\section*{Acknowledgments}
%
%The work of DAZ was supported by the FCT (Portugal) through the grant No. UID/FIS/00618/2013.

\appendix

\section{Proof of Theorem~\ref{prop1}}\label{App_A}

\begin{proof}Consider the following nested intersections
  \begin{align}
  &{\mathcal{V}}_{i^*}\equiv V_{i^*i^*}\cap
  V_{i^*i^*i^*}\cap V_{i^*i^*i^*i^*}\cap \ldots \label{eq:v_cal} \\
  &{\mathcal{H}}_{i^*}\equiv H_{i^*i^*}\cap H_{i^*i^*i^*}\cap
  H_{i^*i^*i^*i^*}\cap\ldots \label{eq:h_cal}
  \end{align}
  Due to Lemma~3.1 from \cite{AA13}, the set ${\mathcal{V}}_{i^*}$ is a
  $v$-curve and the set ${\mathcal{H}}_{i^*}$ is an $h$-curve.
% (see \ref{ProofFig}).
  An intersection of $h$-curve and $v$-curve, ${\mathcal{V}}_{i^*}
  \cap{\mathcal{H}}_{i^*}$, is a point which we denote by ${\mathcal{O}}_{i^*}$.
  Evidently, $T{\mathcal{O}}_{i^*}$ is a point having all its $T$-images and
  $T$-pre-images  situated in $D_{i^*}$. Therefore $T{\mathcal{O}}_{i^*}\in
  {\mathcal{V}}_{i^*}$ and $T{\mathcal{O}}_{i^*}\in {\mathcal{H}}_{i^*}$, which
  implies   that $T{\mathcal{O}}_{i^*}={\mathcal{O}}_{i^*}$.  Since the point
  of  intersection of ${\mathcal{V}}_{i^*}$ and ${\mathcal{H}}_{i^*}$ is unique,
  ${\mathcal{O}}_{i^*}$ is unique fixed point of $T$ in $D_{i^*}$. The
  point (a) is proved.

 Let us prove the point (b).  According to the definition (\ref{eq:v_cal}),
  $p\in {\mathcal{V}}_{i^*}$ if and only if $p\in D_{i^*}$  and $T^{-n}p\in
  D_{i^*}$ for any integer $n>0$. This means that if $p\in {\mathcal{V}}_{i^*}$
  then $T^{-1}p\in {\mathcal{V}}_{i^*}$, therefore ${\mathcal{V}}_{i^*}$
  is invariant with respect to action of $T^{-1}$. If $p\in
  {\mathcal{V}}_{i^*}\subset D_{i^*}$ then $T^{-1}p\in H_{i^*i^*}
  \subset D_{i^*}$, $T^{-2}p\in H_{i^*i^*i^*}\subset D_{i^*}$,
  etc. The $h$-strips $H_{i^*i^*}$, $H_{i^*i^*i^*}$, $\ldots$ form a
  nested sequence
  \begin{displaymath}
    \ldots\subset H_{i^*\ldots i^*}\subset\ldots \subset
    H_{i^*i^*i^*}\subset H_{i^*i^*}\subset D_{i^*}
  \end{displaymath}
  and the intersection in (\ref{eq:h_cal}) is an $h$-curve which
  has the only one intersection point with ${\mathcal{V}}_{i^*}$ at
  ${\mathcal{O}}_{i^*}$.  This implies that if $p\in{\mathcal{V}}_{i^*}$
  then the distance from ${\mathcal{O}}_{i^*}$ to $T^{-n}p\in
  {\mathcal{V}}_{i^*}\cap H_{i^*\ldots i^*}$ ($n$ times $i^*$) can be
  made arbitrarily small by choosing $n$ sufficiently large. Therefore the relation
  (\ref{UnstM}) is proved. The similar reasoning
 allows to prove (\ref{StM}). The point (b)
  is proved.

  Finally, the orbit $\{\ldots, p_{-1},p_0,p_1,\ldots\}$
  corresponding to the code  $\{\ldots,i^*,i^*,i_1,$ $i_2,\ldots\}$,
  $i_1\ne i^*$, has in $D_{i^*}$ infinitely many successive entries
  $p_k$, $k<0$ . All of them are situated in all the strips $V_{i^*i^*}$,
  $V_{i^*i^*i^*}$, etc. Therefore,  $p_k\in {\mathcal{V}}_{i^*}$, $k<0$
  and applying the reasoning of the point (b) one concludes that $\lim_{n\to\infty} p_{-n}={\mathcal{O}}_{i^*}$.  The point (c) is also proved.
\end{proof}

\section{The algorithm for computation of gap solitons}\label{App_D}

%It is assumed that the number of island $2L+1$ is known.
 %and there is a numerical evidence that  Hypotheses 1-3 and Hypothesis+ hold.

\begin{algorithm}[123]
  \caption{Computation of Gap Solitons}
  \begin{algorithmic}[1]

	\State{Fix $M$, $K$ (integer, large enough).
	}
	\label{algo2}

	\State{Compute the eigenvalues $\lambda_1$ and $\lambda_2=1/\lambda_1$,
	$|\lambda_1|>1$, and corresponding eigenvectors, ${\rm e_1}=(\xi,\eta)$,
	${\rm e}_2=(\xi,-\eta)$, of the fixed point ${\mathcal{O}}_0$.
	}

	\Algphase{Phase 1: localization of the strip $H_{0\times M}$
	}

	\State{Solve numerically the Cauchy problem for (\ref{1D_rep}),
	\begin{displaymath}
	x\in[0;M\pi], \quad
	\psi(0)=\varepsilon\xi, \quad
	\psi_x(0)=\varepsilon\eta,
	\end{displaymath}
	for increasing values of $\varepsilon$ spaced by small step $\Delta\varepsilon>0$, 			from $\varepsilon=0$, while $|\psi(M\pi )|<\infty$. Find
	$\varepsilon=\varepsilon^*$ corresponding to the conditions
	\begin{displaymath}
	|\psi(x)|<\infty,
	\quad x\in [0;M\pi)
	\quad \mbox{and}
	\quad \lim_{x\to M\pi}\psi(x)=\infty.
	\end{displaymath}
	}
	
	\State{Assign $r:=\varepsilon^*$ and $l:=-\varepsilon^*$.
	}
	
	%\State{Assign $l:=-\varepsilon^*$.
	%}

	\Algphase{Phase 2: localization of  the strips $H_{(0\times M)i_1}$,
	$H_{(0\times M)i_1i_2}$, $\ldots$, $H_{(0\times M)i_1\ldots i_m}$
	}

    	\State{Allocate $A[0:4L+1]$, an array of $4L+2$ elements.
    	}
    	
    	\For{$k=1$ to $m$}
    	
    	\State\parbox[t]{\dimexpr\linewidth-\algorithmicindent}{
    	Solve numerically the Cauchy problem for (\ref{1D_rep}),
	\begin{displaymath}
	x\in[0;(M+k-1)\pi], \quad
	\psi(0)=r\xi, \quad
	\psi_x(0)=r\eta.
	\end{displaymath}
	\strut}
	
	\State{Assign $s:=\sgn\left\{\lim_{x\to(M+k-1)\pi}\psi(x)\right\}$
	}

	\State\parbox[t]{\dimexpr\linewidth-\algorithmicindent}{
	Solve numerically the Cauchy problem for (\ref{1D_rep})
	\begin{displaymath}
	x\in[0;(M+k)\pi], \quad
	\psi(0)=\varepsilon\xi, \quad
	\psi_x(0)=\varepsilon\eta
	\end{displaymath}
	for increasing values of $\varepsilon$, spaced by small step $\Delta\varepsilon>0$, 			 from $\varepsilon=l$ to $\varepsilon=r$ seeking for
	 $4L+2$ values $\varepsilon=\varepsilon_j^*$, $j=0\div (4L+1)$, such that $\psi(x)$
	satisfies the conditions
	\begin{displaymath}
	|\psi(x)|<\infty,
	\quad x\in [0;(M+k)\pi)
	\quad \mbox{and}
	\quad \lim_{x\to(M+k)\pi}\psi(x)=\infty.
	\end{displaymath}
    Save the found  $4L+2$ values $\varepsilon_j^*$ to the array $A[0:4L+1]$ and sort $A$.
	\strut}

	%\State{Assign $A[j]:=\varepsilon_j^*$.}
	%\State{Sort  $A$.}
	\State{Assign $l:=A[2L+2si_k]$ and $r:=A[2L+2si_k+1]$.}
	%\State{Assign $r:=A[2L+2si_k+1].$}
	
    	\EndFor

	\algstore{algo2_2}
  \end{algorithmic}
\end{algorithm}

\begin{algorithm}
  \begin{algorithmic}[1]
	\algrestore{algo2_2}

	\medskip
	
	\Algphase{Phase 3: localizing $h$-strips from
	$H_{(0\times M)i_1\ldots i_m(0)}$ to
	$H_{(0\times M)i_1\ldots i_m(0\times K)}$
	}
	
	\For{$k=1$ to $K$}

	\State\parbox[t]{\dimexpr\linewidth-\algorithmicindent}{
	Solve numerically the Cauchy problem for (\ref{1D_rep})
	\begin{displaymath}
	x\in[0;(M+m+k)\pi], \quad
	\psi(0)=\varepsilon\xi, \quad
	\psi_x(0)=\varepsilon\eta
	\end{displaymath}
	for increasing values of $\varepsilon$, spaced by small step $\Delta\varepsilon>0$, 			 from $\varepsilon=l$ to $\varepsilon=r$ seeking for
	 $4L+2$ values $\varepsilon=\varepsilon_j^*$, $j=0\div (4L+1)$, such that 				$\psi(x)$ satisfies the conditions
	\begin{displaymath}
	|\psi(x)|<\infty,
	\quad x\in [0;(M+m+k)\pi)
	\quad \mbox{and}
	\quad \lim_{x\to(M+m+k)\pi}\psi(x)=\infty.
	\end{displaymath}
Save the found  $4L+2$ values $\varepsilon_j^*$ to the array $A[0:4L+1]$ and sort $A$.
	\strut}
	
	%\State{Assign $A[j]:=\varepsilon_j^*$.}
	%\State{Sort $A$.}
	\State{Assign $l:=A[2L]$ and  $r:=A[2L+1]$.}
%	\State{Assign $r:=A[2L+1].$}
	
	\EndFor
	
	\Algphase{Phase 4: matching with asymptotics at $x\to+\infty$
	}
	
	\State{Solve numerically the Cauchy problem for (\ref{1D_rep})
	\begin{displaymath}
	x\in[0;(M+m+K)\pi], \quad
	\psi(0)=\varepsilon\xi, \quad
	\psi_x(0)=\varepsilon\eta
	\end{displaymath}
	for increasing values of $\varepsilon$, spaced by small step $\Delta\varepsilon>0$, 			until finding 	 $\varepsilon=\varepsilon^*$ such that $\psi(x)$ satisfies the condition
	\begin{displaymath}
	\frac{\psi_x((M+m+K)\pi)}{\psi((M+m+K)\pi)}=-\frac{\eta}{\xi}
	\end{displaymath}}
  \end{algorithmic}
\end{algorithm}

% You may incorporate your references as follows in your main tex file.
% Using BibTex is not recommended but can be handled.

\medskip
% The data information below will be filled by AIMS editorial staff
Received xxxx 20xx; revised xxxx 20xx.
\medskip

\end{document}